\documentclass[aps,pra,showpacs,amssymb,superscriptaddress,twocolumn,longbibliography]{revtex4-2}
\usepackage[ansinew]{inputenc}
\usepackage{bbm}
\usepackage{bm}
\usepackage{amsbsy}
\usepackage{amsthm}
\usepackage{amssymb}
\usepackage{amsfonts}
\usepackage{upgreek}
\usepackage{amsmath}
\usepackage{dsfont} 
\usepackage{graphicx} 
\usepackage{epsfig}
\usepackage{epstopdf}
\usepackage{svg} 
\usepackage{dsfont}
\usepackage{multibib}
\usepackage{ulem}
\usepackage{thm-restate} 
\usepackage{color}

\usepackage[colorlinks]{hyperref}
\usepackage{orcidlink}
\makeatletter
\newcommand\org@hypertarget{}
\let\org@hypertarget\hypertarget
\renewcommand\hypertarget[2]{%
  \Hy@raisedlink{\org@hypertarget{#1}{}}#2%
  }
\makeatother
\usepackage[figure,table]{hypcap}
\usepackage{MnSymbol}
\usepackage{enumerate}
\usepackage{float}
\hypersetup{
	bookmarksnumbered,
	pdfstartview={FitH},
	citecolor={darkgreen},
	linkcolor={darkred},
	urlcolor={darkblue},
	pdfpagemode={UseOutlines}}
\definecolor{darkgreen}{RGB}{50,190,50}
\definecolor{darkblue}{RGB}{0,0,190}
\definecolor{darkred}{RGB}{238,0,0}
\definecolor{quantum}{RGB}{83,37,127}
\definecolor{quantumlight}{RGB}{169,146,191}
\usepackage{soul}
\newcommand{\pr}{^{\prime}}

\newcommand{\ket}[1]{\ensuremath{\left|\right.\!{#1}\!\left.\right\rangle}}

\newcommand{\bra}[1]{\ensuremath{\left\langle\right.\!{#1}\!\left.\right|}}
\newcommand{\braket}[2]{\ensuremath{\langle{#1}|{#2}\rangle}}
\newcommand{\ketbra}[2]{\ensuremath{|{#1}\rangle\!\langle{#2}|}}

\newcommand{\da}[0]{^{\dagger}}

\newcommand{\subtiny}[3]{\ensuremath{_{\hspace{#1 pt}\protect\raisebox{#2 pt}{\tiny{$ #3$}}}}}
\newcommand{\suptiny}[3]{\ensuremath{^{\hspace{#1 pt}\protect\raisebox{#2 pt}{\tiny{$ #3$}}}}}

\newcommand{\subB}{\ensuremath{_{\hspace{-1pt}\protect\raisebox{0pt}{\tiny{$B$}}}}}

\newcommand{\Poi}{\ensuremath{_{\hspace{-0.5pt}\protect\raisebox{0pt}{\tiny{$P$}}}}}

\newcommand{\trho}{\tilde{\rho}}
\newcommand{\E}{\ensuremath{_{\hspace{-1pt}\protect\raisebox{0pt}{\tiny{$E$}}}}}
\newcommand{\off}[1]{\ensuremath{_{\hspace{-1pt}\protect\raisebox{0pt}{\tiny{$\text{off}$}}}^{\hspace{0pt}\protect\raisebox{0pt}{\tiny{$#1$}}}}}
\newcommand{\gibbs}[1]{\ensuremath{\displaystyle\gamma_{\hspace*{-1.0pt}\protect\raisebox{-1.0pt}{\tiny{$ #1 $}}}}}
\newcommand{\psucc}{\ensuremath{\displaystyle P_{\hspace*{-1.5pt}\protect\raisebox{0pt}{\tiny{succ}}}}}
\newcommand{\Sys}{\ensuremath{_{\hspace{-0.5pt}\protect\raisebox{0pt}{\tiny{$S$}}}}}
\newcommand{\R}{\ensuremath{_{\hspace{-1pt}\protect\raisebox{-1pt}{\tiny{$R$}}}}}
\newcommand{\RS}{\ensuremath{_{\hspace{-1pt}\protect\raisebox{-1.0pt}{\tiny{$R\hspace*{-0.5pt}S$}}}}}
\newcommand{\AB}{\ensuremath{_{\hspace{-1pt}\protect\raisebox{0pt}{\tiny{$A\hspace*{-0.5pt}B$}}}}}
\newcommand{\A}{\ensuremath{_{\hspace{-1pt}\protect\raisebox{0pt}{\tiny{$A$}}}}}

\newcommand{\rank}[1]{\text{rank}\left(#1 \right)}

\newcommand{\tr}{\textnormal{Tr}}

\usepackage{tikz}

\newtheorem{theorem}{Theorem}

\newtheorem{lemma}{Lemma}

\newtheorem{remark}{Remark}

\makeatletter
\renewcommand{\p@subsection}{}
\renewcommand{\p@subsubsection}{}
\makeatother

\newcommand{\M}[1]{\mathcal{#1}}

\newcommand{\id}{\mathbb{I}}

\newcommand{\cmax}{C\subtiny{0}{0}{\mathrm{max}}}

\usepackage{paracol}
\usepackage{tcolorbox}
\tcbuselibrary{theorems}

\begin{document}

\title{Thermodynamic Constraints on Information Transmission in Quantum Ensembles}

\author{Andr\'e T. Ces\'{a}rio\,\texorpdfstring{\orcidlink{0000-0002-6972-2576}}{}}
\affiliation{Departamento de F\'{\i}sica - ICEx - Universidade Federal de Minas Gerais,
Av. Pres. Ant\^onio Carlos 6627 - Belo Horizonte - MG - Brazil - 31270-901.}
\author{Tiago Debarba\,\texorpdfstring{\orcidlink{0000-0001-6411-3723}}{}}
\email{debarba@utfpr.edu.br}
\affiliation{Departamento Acad{\^ e}mico de Ci{\^ e}ncias da Natureza - Universidade Tecnol{\'o}gica Federal do Paran{\'a}, Campus Corn{\'e}lio Proc{\'o}pio - Paran{\'a} -  86300-000 - Brazil}

\begin{abstract}
The processing of quantum information is limited by fundamental physical constraints on how information can be encoded, transmitted, and extracted. In particular, the non-orthogonality of quantum states limits their distinguishability, and thermodynamic constraints, including the energetic cost of state preparation and quantum operations, further restrict the viability of realistic information protocols. This work explores the impact of such constraints on the preparation, evolution, and readout of quantum information. We demonstrate that preparing the system for encoding and measurement affects the distinguishability and purity of the resulting ensemble of states. Furthermore, we analyze a noisy communication channel and propose an optimal protocol for encoding and decoding the transmitted information. For this realistic protocol, we show that the maximum probability of successfully retrieving the information is equal to the maximum correlation that can be achieved between the system and the register.~The protocol uses only Gibbs states as free resources, ensuring minimal thermodynamic cost.~Based on this, we provide a thermodynamic interpretation of the Holevo information, which quantifies the capacity of the transmitted information and establishes a fundamental limit on its retrieval in thermodynamically constrained scenarios.
\end{abstract}

\maketitle

\section{Introduction}

Information is a physical quantity that must be encoded in a physical system to be processed, transferred, and measured \cite{landauer1961,Anderson2017}.~Consequently, it is subject to the physical laws that govern the system in which it is encoded. When information is stored in microscopic systems, it adheres to the principles of quantum physics. Phenomena such as coherence and non-locality enable intrinsic quantum protocols that have no classical counterparts, including quantum cryptography \cite{BennettBrassard84,Ekert91,bennett1992,GisinMassar1997} and quantum teleportation \cite{BennettEtAl93}. On the other hand, decoherence can rapidly degrade essential correlations \cite{zurek1991decoherence}. When thermodynamic principles are considered, thermal noise further challenges the long-term preservation of coherence \cite{xuerebprl}. For example, the {\it third law of thermodynamics} imposes fundamental constraints on the attainable purity of states in microscopic systems \cite{masanes2017general,wilming2017third,Taranto2023,taranto2025efficiently}, and creating or operating over coherences \cite{lostaglio15,lostaglio19, Gour22}. 
\begin{figure}[t]
    \centering
    \includegraphics[width=1\linewidth]{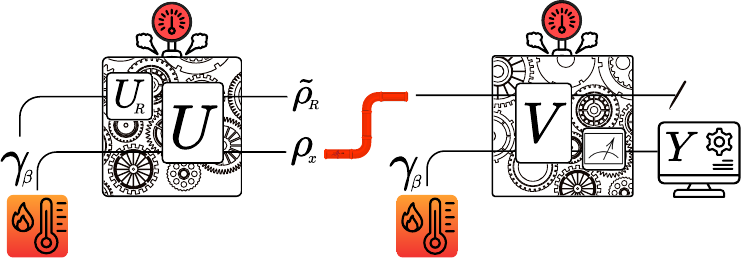}
     \caption{{\bf Thermodynamic Coding/Decoding Protocol:} Thermal states denoted by $\gibbs{\beta}$ are freely accessible through interaction with a thermal bath at inverse temperature $\beta$ and Hamiltonian $H$. The {\it sender} prepares a register state given by $\rho\R$ by applying unitary operations $U_R$ to a Gibbs state $\gibbs{\beta}$. The information $X$, encoded in $\rho\R$, is transmitted to the {\it receiver} via a unitary interaction $U$ among the register and a specified number of copies of the Gibbs state $\gibbs{\beta}$. The {\it receiver} attempts to decode the information by analyzing the statistical properties of the received states $\rho_x$.~The measurement process is inherently non-ideal, as the measurement pointer is also initialized in a Gibbs state $\gibbs{\beta}$. Heat is generated in proportion to the amount of encoded and decoded information throughout the encoding and decoding process.}
    \label{fig: protocol}
\end{figure}
 
Since unknown quantum states cannot be cloned \cite{wooterzurek82,Werner1998}, the information must be transmitted from one part to another. Even unknown classical information is subject to fundamental limitations on broadcasting \cite{CalevHen2008,daffertshofer2002classical}.  When thermodynamic constraints are considered, classical information cannot be fully transmitted to more than a single \textit{receiver} \cite{debarba2024}. Consequently, encoding classical information in quantum systems is advantageous because quantum coherence enhances the security of the transmitted information.~Non-orthogonality among quantum states introduces indistinguishability, which makes it more difficult for a malicious eavesdropper to intercept and extract the encoded information \cite{BennettBrassard84}. The \textit{sender} and \textit{receiver} must establish a predefined strategy for encoding and decoding information, to ensure that only the intended parties can access it. However, how should this strategy be designed when {thermal} noise affects the preparation, transmission, and measurement processes? Motivated by this question and its challenge, we explore the fundamental limitations on information encoding in quantum systems imposed by finite thermodynamic resources. These constraints are intrinsically linked to the third law of thermodynamics \cite{Taranto2023}, which governs the achievable purity of state preparation \cite{taranto2025efficiently} and the precision of measurement processes \cite{GuryanovaFriisHuber2018}. 

Recent advances have emphasized the increasing significance of thermodynamic considerations in quantum information processing. For instance, the thermodynamic cost of erasing information in quantum error correction processes (QEC) such as ancilla resets has been shown to cause heating, which is governed by Landauer's principle. This heating could potentially trigger dynamical phase transitions in the feasibility of fault-tolerant quantum computing~\cite{qec_landauer,Li15,Fellous21}. Conversely, QEC protocols have been modeled as autonomous thermal machines under feedback control, revealing energetic trade-offs that constrain information recovery~\cite{Danageozian22}.~Additionally, studies of small-scale thermodynamic irreversibility have demonstrated that coherence and finite-size effects can drive nontrivial entropy production~\cite{irreversibility_coherence}.~Indeed, recent surveys have shown how resource constraints, including thermodynamics, may influence the design of early fault-tolerant architectures~\cite{zhang2025,Kshirsagar2024}.~In Ref.~\cite{Biswas22}, the optimal encoding of classical information into quantum systems under constraints of finite thermodynamic resources is thoroughly analyzed. In that work, second-order asymptotic bounds were derived relating the achievable encoding rate to fluctuations in free energy.~Our approach adopts an ensemble-based operational perspective, focusing on the structural constraints imposed by the number of states available, the rank, and distinguishability. 

This work examines noisy communication channels where the noise originates from fundamental thermodynamic constraints.~We consider a scenario in which a \textit{sender} encodes a random variable, or message, into an ensemble of quantum states and transmits it to a \textit{receiver}, who aims to decode the information. Our framework involves preparing a register state and communicating through a finite-dimensional thermal reservoir at a finite temperature. Interactions are governed by unitary dynamics, as illustrated in Fig.~\ref{fig: protocol}.~The quantum communication channel becomes inherently noisy when the quantum states encoding the information are mixed \cite{schumacher1997sending}, thereby restricting the amount of information that can be compressed by the \textit{sender} \cite{horodecki1998limits} and retrieved by the \textit{receiver} \cite{Holevo1973}.  

Alternatively, Ref.~\cite{chiribella22} examines the inherent cost of using a non-equilibrium state for quantum information processing. This study advances the discussion by avoiding the explicit attribution of an energy cost to preparing the auxiliary system.~We assume that only thermal states are accessible and acknowledge that thermodynamic constraints are inherent and unavoidable in any realistic setting.~Unlike approaches based on catalytic or single-shot transformations, our analysis focuses on constraints at the ensemble level arising from protocols that can be physically implemented and are governed by thermal operations.~For example, Ref.~\cite{Narasimhachar2019} introduces memory capacity as a thermodynamic resource and analyzes how limited access to purification restricts quantum information tasks under low-temperature conditions. Refs.~\cite{Hsieh2025PRL, Hsieh2025PRA} develop a dynamic version of Landauer's principle, establishing a quantitative equivalence between the energy cost and the amount of classical information transmitted. These studies emphasize energetic trade-offs and transition feasibility between quantum states. Our framework, in contrast, investigates how thermodynamic restrictions shape the structure of ensembles prepared from Gibbs states, as well as how structural features such as rank subadditivity and limited distinguishability impose fundamental limitations on decoding success. 

In a scenario where the \textit{sender} and \textit{receiver} have access only to thermal states, it is demonstrated that the compression rate of an ensemble of pure states remains unaffected by these constraints. However, the fidelity of the decoded information is reduced \cite{Xuereb_2025}. According to the Schr\"{o}dinger-HJW theorems \cite{schrodinger1935mixture,HJW93}, any quantum ensemble can be generated from a given density matrix via unitary operations with rank-deficient orbits. This process also requires preparing and transforming pure states, which becomes infeasible when considering thermodynamic constraints. We circumvent this limitation by proposing a protocol that encodes information in an ensemble of quantum states through the interaction of resourceless thermal states \cite{Brandao13, Horodecki13}.~We thus present an optimal decoding protocol that maximizes the distinguishability of the states in the ensemble based on the preparation of ancillary thermal states.~The following assumptions underpin our proposed protocol: (1) only thermal states can be prepared, (2) a finite size thermal bath at finite temperature is freely accessible to both the \textit{sender} and the \textit{receiver}, and (3) the interaction and evolution operations are unitary. 

The main contributions of this work are as follows: (i) We demonstrate that monotonicity of rank implies thermodynamic constraints that limit state preparation for a given ensemble, (see Lemma~\ref{lemma: main rank inequality}); and (ii) We prove that encoding classical information into an ensemble of non-orthogonal pure states is impossible under such constraints (see Theorem~\ref{theorem: non-orthogonal ensembles}) and that any ensemble of $n$ linearly dependent $d\Sys$-dimensional states (with $n > d\Sys$) must consist of full-rank elements (see Theorem~\ref{theorem: mixture_theorem}). (iii) We present an optimal encoding and decoding protocol based on controlled unitaries.~This protocol shows that the success probability of distinguishing the encoded states is equal to the maximum achievable statistical correlation  (see Theorem~\ref{theorem optimal ensemble}). (iv) We provide a thermodynamic interpretation of Holevo information, showing that it is bounded above by the heat generated during the encoding process and proportional to the entropy difference relative to the Gibbs state.~Together, these results establish fundamental thermodynamic limits on preparing, transmitting, and recovering classical information in quantum communication settings.
 
This work is structured as follows: Section~\ref{sec: framework} introduces the quantum information communication framework and details the thermodynamic constraints and encoding protocol. Section~\ref{sec: main theorems} explores the limitations that thermodynamic noise imposes on the rank of encoded states.~Section~\ref{sec: optimal encoding} presents an optimal encoding/decoding protocol based on controlled operations.~Section~\ref{sec: holevo and heat} analyzes the thermodynamic relationships governing the amount of encoded information. Finally, Section~\ref{sec: conclusion} summarizes our findings and outlines future perspectives.

\section{Framework}\label{sec: framework}
\begin{figure*}[t]
    \centering
    \includegraphics[width=1.0\linewidth]{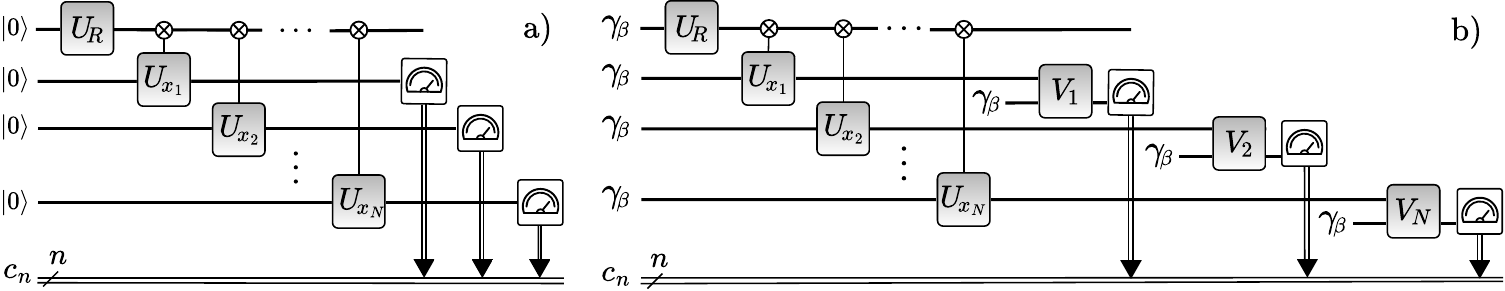}
    \caption{\textbf{Communication protocols}: The sketch shows a quantum communication protocol that is repeated $N$ times. In this protocol, the \textit{sender} prepares a register state by acting on a random unitary $U\R$. Then, in each $i-th$ round, the \textit{sender} encodes the information into the transmitted state through a controlled operation
     $U=\sum_{x_i}\ketbra{x_i}{x_i}\otimes U_{x_i}$, where $U_{x_i}$ indicates the action over the transmitted system. The \textit{receiver} then measures its apparatus and records the result on a $n$-dimensional classical tape, $c_n$. Fig. $2a)$ represents a perfect communication protocol, in which the \textit{sender} can prepare the register and the transmitted system at pure state $\ket{0}$. The \textit{receiver} can perform a perfect measurement. 
    Fig.~$2b)$ represents the noisy version of the protocol. In this version, {the \textit{sender} {prepares} the register and the transmitted system,} while the \textit{receiver} {prepares}  the measurement apparatus at a thermal state, $\gibbs{\beta}$, in contact with a thermal bath at an inverse temperature, $\beta$. }
    \label{fig: circuit}
\end{figure*}

A quantum communication protocol involves a \textit{sender} $A$ and a \textit{receiver} $B$, where   $A$ conveys information using a quantum channel  $\Lambda$ to transmit to $B$. Since information is a physical quantity \cite{landauer1961,Anderson2017}, the \textit{sender} $A$ must encode {the} information in a physical system. If the physical system is quantum, $A$ can transmit either quantum or classical information. This channel can be represented by the shared output state $\rho\AB$, and the amount of information that can be transferred from $A$ to $B$ is quantified by the mutual information of $\rho\AB$
\begin{eqnarray}\label{eq: mutual_info}
    I(A\!:\!B)_{\rho\AB} = S(\rho\A) + S(\rho\subB) - S(\rho\AB),
\end{eqnarray}
where $S(\rho)= - \tr(\rho\log\rho)$ is the von Neumann entropy of a given state $\rho$. 
If $\rho\AB$ is maximally entangled, $A$ and $B$ can use its non-locality to perform fully quantum information protocols, with no classical counterparts, as, for example, in quantum cryptography protocols \cite{BennettBrassard84,Ekert91,GisinMassar1997} or in quantum teleportation \cite{BennettEtAl93}. Moreover, $A$ can transmit a classical message $X$, encoded by a set of letters $\{x\}_{x=0}^{n-1}$,  with cardinality $|X| = n$. This message can be produced by an information source in which each letter is sampled by the random variable $X = \{x,\,p_x\}$, where $p_x$ characterizes the probability of generating a letter $x$ \cite{preskill16a}. 
For a quantum source, the information will be encoded in an ensemble of quantum states, denoted by $\Sigma = \{p_x,\,\rho_x\}_{x=0}^{n-1}$. On average, the state of the quantum system is described by the density matrix $ \trho\Sys = \sum_{x=0}^{n-1} p_x \rho_x $. Since the message is a classical entity, the quantum channel is represented as a process that converts a classical register $R$ into a quantum state $\rho\RS$, as
\begin{eqnarray}\label{eq: class_quantum}
    \rho\RS = \sum_x p_x \ketbra{x}{x}\R\otimes \rho_x,
\end{eqnarray}
where $R$ is an $n$-dimensional partition of $A$ and $S$ is a $d\Sys-$dimensional partition to be sent to $B$. As previously mentioned in Eq.~\eqref{eq: mutual_info}, the mutual information $I(R\!:\!S)_{\rho\RS}$ quantifies the amount of classical information that can be encoded in the ensemble $\Sigma$, represented by the so-called \textit{Holevo Information} \cite{Holevo1973} 
\begin{eqnarray}\label{eq: holevo_inf}
    \chi(\Sigma) = S\left(\sum_x p_x \rho_x\right) - \sum_x p_x S(\rho_x), 
\end{eqnarray}
with $I(R\!:\!S) = \chi(\Sigma)$ for $\rho\RS$ as presented in Eq.~\eqref{eq: class_quantum}. 
{If the ensemble is composed of non-orthogonal pure states $\ket{\psi_x}$, the Holevo information corresponds to the von Neumann entropy of the system $\rho\Sys$, \textit{i.e.,} $\chi(\Sigma)=~S(\sum_x p_x\ketbra{\psi_x}{\psi_x})$. Moreover, if the states $\rho_x$ are pure and orthogonal, $A$ is solely communicating classical information, as $\rho\RS$ exhibits only classical correlations, and the Holevo information attains the Shannon entropy of $X$, \textit{i.e,} $\chi(\Sigma)= H(X) = -\sum_x p_x\log p_x$ \cite{OllivierZurek01}.} 
If the states $\rho_x$ are mixed, the communication channel is considered noisy, as the pure states $\ketbra{\psi_x}{\psi_x}$ are mapped to mixed states $\rho_x$, see Ref. \cite{schumacher1997sending}. 

The probability of success in distinguishing the states of the ensemble is quantified as
\begin{equation}\label{eq: prob_success}
    \psucc = \max_{\substack{0\leq P_x\leq \id, \\ \sum_x P_x =\id}}\left(\sum_x p_x\tr(\rho_xP_x)\right),
\end{equation}
where $P_x$ are the elements of a \textit{Positive Operator-Valued Measure} (POVM). If the states on the ensemble are distinguishable, {there} exists a POVM that can perfectly discriminate them and $\psucc=1$ \cite{ivanovic1987,Chefles2000,Roa2002}.

\subsection{Noiseless Communication Protocol}

A transformation $\Lambda:\M{B}(\M{H})\rightarrow\M{B}(\M{H}')$ that maps a quantum state into another is called a completely positive and trace preserving (CPTP) map \cite{kraus1971}.~In general,~CPTP maps can be expressed as a unitary interaction between the state being transformed and an auxiliary system $E$ prepared in a known pure state, represented as $\ket{0}$, \cite{choi1975} 
\begin{eqnarray}\label{eq: cptp dilation}
    \Lambda(\rho) = \tr\subtiny{0}{0}{E\pr}(U \rho\otimes\ketbra{0}{0} U^{\dagger}),
\end{eqnarray}
where $\M{H}\subtiny{0}{0}{E\pr}\subseteq \M{H}\otimes\M{H}\E$ is a fraction of the entire system {plus} environment, dependent on the output Hilbert space $\M{H}\pr$. To encode information into the ensemble $\Sigma$, the map $\Lambda$ should operate over a diagonal state $\rho\R = \sum_xp_x\ketbra{x}{x}$, resulting $\rho\RS$ in Eq.~\eqref{eq: class_quantum}.~There exists a so-called controlled unitary $U = \sum_x \ketbra{x}{x}\otimes U_x$, with set $\{U_x\}$ composed of unitaries, resulting  
\begin{eqnarray}\label{eq: ideal dilation ensemble}
    \tr\subtiny{0}{0}{E\pr}(U \rho\R\otimes\ketbra{0}{0}\E U^{\dagger}) =  \sum_x p_x \ketbra{x}{x}\otimes \rho_x,
\end{eqnarray}
with $\rho_x \equiv \tr\subtiny{0}{0}{S\pr}(U_x\ketbra{0}{0}U_x^{\dagger})\in\M{B}(\M{H}\Sys)$, where $\M{H}\subtiny{0}{0}{S\pr}\subseteq \M{H}\Sys\otimes\M{H}\E$ \cite{HJW93}. 

Fig.~\ref{fig: circuit}$a)$ illustrates the discussed protocol.~A register state encodes the random variable  $X$ in its diagonal, represented by the density matrix $\rho\R = \sum_x p_x\ketbra{x}{x}$, under the action of a random unitary operation $U\R$ over a pure state $\ket{0}$. This register state interacts with another system prepared in the known pure state  $\ket{0}$ via a controlled unitary interaction $U$, to obtain the state presented in Eq.~\eqref{eq: ideal dilation ensemble}.~The system is then transmitted to the \textit{receiver}, who measures it to distinguish the states and retrieve the information $X$.

From the above expression Eq.~\eqref{eq: ideal dilation ensemble}, it is evident that any desired ensemble can be achieved with sufficient knowledge and precise control.~However, what limitations exist on ensembles encoding a random variable $X$ when systems $A$ and $B$ are governed by thermodynamic laws? More specifically, what constraints arise when only finite thermodynamic resources are available, such as access to mixed states, thermal baths at finite temperatures, coherence-preserving operations, or a limited ability to generate quantum correlations?

\subsection{Noisy Communication Protocol}\label{sec:noisyprotocol}

The laws of thermodynamics impose restrictions on the preparation \cite{Taranto2023}, operation \cite{FaistRenner2018}, and measurement process \cite{GuryanovaFriisHuber2018} of quantum systems. Specifically, the third law states that only full-rank states can be prepared with a finite amount of thermodynamic resources  \cite{wilming2017third}. 
Considering that the auxiliary system is in contact with a thermal bath and thermalizes to a Gibbs state at temperature $T = 1/\beta$ (considering $k_B = 1$), with finite-dimensional Hamiltonian $H= \sum_{i}E_i \ketbra{E_i}{E_i}$ 
\begin{equation}\label{eq: thermal state}
    \gibbs{\beta} = \frac{\exp{(-\beta H)}}{Z},
\end{equation}
where $ Z = \tr(\exp(-\beta H)) $ is the partition function.  Consequently, in thermodynamic scenarios, the set of CPTP maps over states is usually limited to those that preserve the Gibbs state. We can define a finite-resource process, dependent on $ \gibbs{\beta} $ and {$\Lambda\subtiny{0}{0}{\beta}$} over a given state $ \rho $, as follows:
\begin{eqnarray}\label{eq: thermo cptp dilation}
    \Lambda\subtiny{0}{0}{\beta}(\rho) = \tr\subtiny{0}{0}{E\pr}(U \rho\otimes\gibbs{\beta} U^{\dagger}).
\end{eqnarray}

The transformation, denoted as $\Lambda\subtiny{0}{0}{\beta}(\rho)$, satisfies what we call a finite resource protocol and can be interpreted as a thermodynamically noisy version of Eq.~\eqref{eq: cptp dilation}.~Since our focus is on the thermodynamic aspects of information encoding, we consider a model in which the encoding process is implemented through a global unitary interaction between the information register, represented by the density matrix $\rho\R$ and a thermal auxiliary system, initialized in the Gibbs state $\rho\Sys = \gamma_\beta $.~The partial trace will be taken to obtain the desired post-interaction states. Hence, the post-interaction state $\trho\RS$ is given by the action of the unitary $U$, as follows
\begin{equation}\label{eq: unitary interaction}
    \trho\RS = U \left(\rho\R\otimes\rho\Sys\right) U^{\dagger}.
\end{equation}
Note that this assumption is not restrictive, even though the protocol is described using unitary operations. Within the framework of thermodynamic resource theories \cite{Brandao13, Horodecki13}, Gibbs states are considered free, and introducing ancillary systems in thermal states incurs no cost. Including such ancillas in the global unitary and tracing them out allows for the realization of more general quantum operations (\textit{i.e.}, CPTP maps), as guaranteed by the Stinespring dilation theorem (see \cite{Stinespring55}).~Consequently, the formalism naturally extends to open-system dynamics built from thermodynamically free resources.

Fig.~\ref{fig: circuit}$b)$ illustrates the thermodynamic noisy encoding process, as detailed in App.~\ref{app_sec: encoding decoding protocol}. In summary, it consists of: 
\begin{enumerate}
    \item[(\ref{app_subsubsec: register_preparation})] \textbf{Encoding the classical variable on $\rho\R$}: The \textit{sender} has free access to thermal states $\gibbs{\beta}$ and prepares the register by encoding the classical random variable $X$ in its diagonal through the application of a Haar-random unitary $ U\R $. The resulting state is given by $\rho\R = U\R \gibbs{\beta} U\R^{\dagger}$, in the form
    \begin{equation}
    \rho\R = \sum_{x=0}^{n-1} p_x \ketbra{x}{x} + \rho^{\text{off}}\R,
    \end{equation}
    Here, the off-diagonal terms are denoted by $ \rho^{\text{off}}\R=\sum_{x\neq y}\bra{x}\rho\R\ket{y}$.
\item[(\ref{app_subsubsec: system_preparation})]\textbf{Preparing the system state $\rho\Sys$}:
    The system is prepared in a thermal state, defined as $\rho\Sys =\gibbs{\beta}$. For a random variable with cardinality $n$, the system's Hilbert space can be decomposed into $n$ subspaces, such that $\M{H}\Sys = \M{H}_0 \oplus \cdots \oplus \M{H}_{n-1}$, where each letter $x$ of $X$ is encoded in a corresponding $d_x-$dimensional subspace given by $\M{H}_x$.~This partitioning of the Hilbert space satisfies the condition $d\Sys = \sum_{x=0}^{n-1} d_x$.
\item[(\ref{app_subsubsec: encoding})]\textbf{Transferring the information}
    The transfer and encoding of information from the register to the system occurs through a controlled unitary interaction of the following form 
    \begin{equation}
    U = \sum_x \ketbra{x}{x}\R \otimes U_x,
    \end{equation}
    thereby generating the ensemble $ \Sigma = \{p_x, \rho_x\} $ in the system, in a given measurement basis. 
 \item[(\ref{app_subsubsec: decoding})]\textbf{Decoding the information}
    The \textit{receiver} decodes the information by measuring the average state $ \trho\Sys = \sum_x p_x \rho_x $. This is achieved by coupling $ \rho_x $ to a pointer system, initialized in the thermal state $ \gibbs{\beta} $, through an interaction $V$ that maximizes $\psucc$ as given in Eq.\,\eqref{eq: prob_success}.
\end{enumerate}
Since the diagonal of the register state $ \rho\R$ remains unchanged throughout the protocol, the process can be repeated as many times as necessary for reliable transmission and recovery of information.

In Section~\ref{sec: main theorems}, we present the limitations of preparing the ensemble $\Sigma$ based on the unitary encoding as described in Eq.~\eqref{eq: unitary interaction}. This is our only assumption for now. In Section~\ref{sec: optimal encoding}, we discuss in detail the ``thermodynamic-friendly" encoding and decoding strategies that maximize the distinguishability of the $\rho_x$, based on the above-described protocol.

\section{Encoding limitations using finite thermodynamics resources} \label{sec: main theorems}

In this section, we show how restrictions on the rank of the input state affect the ensemble generation process. Specifically, the product of the ranks of the system and the register cannot exceed a scaled version of the ensemble's maximum rank. This implies that increasing the number of available states in the ensemble or enhancing their distinguishability requires additional resources in the system or the register.~We derive fundamental bounds on the rank of states within an ensemble and prove Theorem~\ref{theorem: non-orthogonal ensembles}, which states that it is impossible to encode a classical random variable $X$ into a set of pure quantum states using finite thermodynamic resources.  Theorem~\ref{theorem: mixture_theorem} formalizes the rank condition that must be satisfied when the ensemble consists of linearly dependent states. This theorem establishes that under finite resources, the states in the ensemble must all have the same rank. Specifically, if the states are linearly dependent, then each must have rank equal to the dimension of the system, ensuring that all states in the ensemble are full-rank. This result imposes strict conditions on preparing quantum ensembles and highlights thermodynamic limitations when encoding classical information into quantum states under these constraints.

The proof of Lemma~\ref{lemma: main rank inequality} relies on the well-known sub-additivity property of rank, which states that the rank of a bipartite quantum state cannot exceed the product of the ranks of its marginal states (see Ref.~\cite{Cadney2013}). While the mathematical structure of our result is related to this property, its operational meaning is novel in this context.

\begin{lemma}[Ensemble Rank Relation]\label{lemma: main rank inequality}
Consider a register state $\rho\R$, which encodes a random variable $X$ with cardinality $n$ in its diagonal elements $p_x$. This register unitarily interacts with a system prepared in the state $\rho\Sys$. If the post-interaction state of the system is given by the averaged state $\trho\Sys = \sum_{x=0}^{n-1} p_x \rho_x$, the rank of the states $\rho_x$ will satisfy
\begin{equation}\label{eq: main rank inequality}
  \text{rank}(\rho\Sys) \cdot \text{rank}(\rho\R) \leq n \cdot \text{rank}(\rho_{\text{max}}),
\end{equation}
where $\rho_{\text{max}}$ is the state with the highest rank within the set $\{\rho_x\}_{x=0}^{n-1}$.
\end{lemma}

\noindent \textit{Proof sketch}: The total rank is invariant under unitaries and $\text{rank}(\sum_{i=0}^{n-1}p_x\rho_x)\leq n~ \text{rank}(\rho_{\text{max}})$. The complete proof is presented in Appendix~\ref{app_sec: encoding limitations}.\\

Thus, Lemma~\ref{lemma: main rank inequality} translates a structural mathematical relationship into a thermodynamic constraint on information encoding under finite-resource assumptions. Therefore, the rank limitation imposed by the third law of thermodynamics shows that it is impossible to encode classical information into a set of non-orthogonal pure states. This result follows directly from the rank constraint: if all states in the ensemble are pure, that is, $\rho_x = \ketbra{\phi_x}{\phi_x}$, then each of them necessarily has rank one. This leads to a contradiction with Lemma \ref{lemma: main rank inequality} when both the system and the register are full-rank states.
 Physically, this implies that any ensemble encoding under finite resources must inherently involve mixed states. This conclusion is consistent with well-established results in quantum thermodynamics, where the preparation of pure states typically incurs an infinite thermodynamic cost, as the third law of thermodynamics dictates  \cite{Taranto2023}. Next, Theorem~\ref{theorem: non-orthogonal ensembles} states the informational limitations imposed by Lemma~\ref{lemma: main rank inequality}.

\begin{theorem}[No-go pure states ensemble]\label{theorem: non-orthogonal ensembles}
Considering the scenario proposed in Lemma \ref{lemma: main rank inequality}, there is no finite resource (both system and register must be full-rank) protocol capable of encoding a random variable $X$ into an ensemble of non-orthogonal pure states.
\end{theorem}
\begin{proof}
   From Lemma \ref{lemma: main rank inequality}, if the states in the ensemble are pure, i.e., $\rho_x = \ketbra{\phi_x}{\phi_x}$ for all $x \in \{0,\ldots,d_{\text{R}}-1\}$, the rank of $\rho_{\text{max}}$ is 1. Consequently, the rank of the system state must satisfy    
    \begin{equation}
        \rank{\rho\Sys} \leq \frac{n}{\rank{\rho\R}} = 1.
    \end{equation}
Thus, the system must be prepared in a pure state. This contradicts the assumption of finite resources due to the third law of thermodynamics.
\end{proof}

This no-go result directly illustrates the third law of thermodynamics in the context of quantum information processing. As demonstrated in Ref.~\cite{Taranto2023}, creating pure states from a thermal state at a finite temperature requires the use of infinite thermodynamic resources. Our theorem formalizes this unattainability within a communication protocol. When both the register and the system are restricted to full-rank thermal states, any attempt to generate an ensemble of non-orthogonal pure states will fail due to rank limitations. This operationalizes the third-law constraint as a structural limit on information encoding, whereby only mixed-state ensembles can be prepared with finite resources. Thus, our result connects the resource-theoretic interpretation of the third law to particular restrictions on preparing informationally rich quantum ensembles.

In Ref.~\cite{schrodinger1935mixture}, Schr\"odinger, in its well-known {\it mixture theorem}, claims that a given $d\Sys$-dimensional diagonal density matrix $\rho=\sum_l \lambda_l \ketbra{\lambda_l}{\lambda_l}$ can always be be written as a mixture of pure states $\rho\Sys\pr =\sum_{k=0}^{n-1}p_k,\ketbra{\psi_k}{\psi_k}$, for any finite $n\geq d\Sys$, if there exist a logical unitary operation $U$ such that $p_k\ket{\psi_k}=\sum_l u_{k,l}\sqrt{\lambda_l}\ket{\lambda_l}$, where $u_{k,l}$ are the matrix elements of $U$. Independently, in Ref.~\cite{HJW93}, the so-called {\it HJW-Theorem} states that there is always a unitary $U$ in this form, with $n-d\Sys$ zero columns. They also show that the ensemble can be obtained by the local action of $U$ on the purified state $\ket{\rho}=\sum_l \sqrt{\lambda_l} \ket{\lambda_l}\ket{\lambda_l}$. Moreover, the unitary $U$ does not act on the Hilbert space, it acts over the list of $\sqrt{\lambda_l} \ket{\lambda_l}$ to obtain another list $\sqrt{p_k},\ket{\psi_k}$, which cannot be experimentally implemented. Also, the unitary orbit is not full-rank because it has a kernel subspace \cite{zycbook}. This implies that the action of those unitaries on full-rank states would result in purified subspaces, which violates the premise of the third law of thermodynamics. Therefore, when the ensemble consists of linearly dependent states, with $n > d\Sys$, and is subject to finite resource constraints, it must satisfy a strict rank condition. The following theorem formally expresses this condition and characterizes the rank of the states in the ensemble. It establishes that, under the given constraints, the states must all be of full rank.

\begin{theorem}[Thermodynamic Mixture Theorem for Linearly Dependent States]\label{theorem: mixture_theorem}
   Under finite resources, if the states of the ensemble 
   $\Sigma = \{p_x,\rho_x\}_{x=0}^{n-1}$ are linearly dependent, they share the same rank. Therefore, they must all be full-rank  
   \begin{equation}
   \rank{\rho_x} = d\Sys, \quad \forall\, x \in \{0,\ldots,n-1\},
   \end{equation} 
   where $d\Sys$ is the dimension of ${\rho_x}$.
\end{theorem}
\begin{proof}
Given that we assume the system is initially prepared in a full-rank state, the rank of the average state must be full-rank since the ensemble consists of linearly dependent states. Thus, we have:
\begin{equation}
    \rank{\sum_x p_x \rho_x} = \sum_x p_x \rank{\rho_x}.
\end{equation}
On the other hand, from Lemma~\ref{lemma: main rank inequality}, we obtain the following bound:
\begin{equation}
    \rank{\rho_x} \geq \rank{\rho\R} \frac{d\Sys}{n}.
\end{equation}
Since we assume $d\R \geq n$, it follows that
\begin{equation}
    \rank{\rho_x} \geq d\Sys.
\end{equation}
Concluding the proof. 
\end{proof}

We refer to this result as a \textit{thermodynamic mixture theorem for linearly dependent states}, which can be interpreted as a \textit{noisy version} of the Schr\"{o}dinger-HJW Mixture Theorem~\cite{schrodinger1935mixture,HJW93}. It describes how, in the presence of noise and finite resources, different convex decompositions of the same mixed state become operationally indistinguishable. Unlike the original mixture theorem, which assumes the preparation of pure states and allows the action of orbit unitaries on supports of arbitrary rank, the thermodynamic mixture theorem presented here does not rely on the preparation of pure states nor the availability of unitaries acting on non-full-rank supports. These restrictions arise naturally from the thermodynamic limitations associated with finite resources.  In particular, unitary operations that act on restricted-support (i.e., non-full-rank) states cannot be accessed without exceeding the finiteness of resources, as they can purify mixed states.

Moreover, in the special case where the ensemble consists of orthogonal pure states, the thermodynamic constraints manifest themselves through the rank of the average state $\trho\Sys$, which quantifies the maximal statistical correlations accessible within the allowed energy and entropy budgets.

\begin{remark}[{Orthogonal pure states}]\label{remark orthogonal pure states}
If the ensemble is composed by orthogonal pure states $\Sigma = \{p_x,\ketbra{\phi_x}{\phi_x}\}_{x=0}^{n-1}$, resulting $\trho\Sys = \sum_x p_x \ketbra{\phi_x}{\phi_x}$. In this case $\rank{\rho_x} = 1$, and $\trho\Sys = V\rho\Sys V\da$, implying that 
    \begin{equation}
        \rank{\trho\Sys} = \rank{\rho\Sys} = n \leq d\Sys.
    \end{equation}
\end{remark}

Unlike in the non-orthogonal case, where indistinguishability results from rank constraints, orthogonal pure states allow for the complete retrieval of all encoded information in the ensemble without thermodynamic constraints. This is because orthogonal states are perfectly distinguishable, enabling the Holevo information quantity to reach its maximum value, which corresponds to the Shannon entropy of the classical variable $X$.~Consequently, orthogonal ensembles achieve maximum distinguishability and optimal information encoding efficiency while remaining consistent with the constraints imposed by finite thermodynamic resources.

\section{Encoding and decoding information with controlled operations} \label{sec: optimal encoding}

The initialization of the register state incurs a cost in thermodynamic resources, for example, through the action of a finite resource operation as described in Eq.~\eqref{eq: thermo cptp dilation}.~Preparing the register state can be time-consuming and energetically costly, and may introduce unaccounted sources of errors. As presented in Section\,\,\ref{sec:noisyprotocol}, here, we consider that the preparation of the register state is achieved by applying a random unitary operation, designated as $U\R$. This operation encodes the random variable $X$ in the diagonal of the register state. {Note that this process cannot be repeated because $U\R$ is
sampled at random. Here, the register is prepared once by sampling a Haar-random $U_R$ that encodes X in its diagonal. The sampling of $U\R$ is a one-off step that is not reused; however, the subsequent controlled-$U$ interactions preserve the register diagonal and can be repeated arbitrarily for multiple rounds of transmission and measurement.} Inevitably, the \textit{sender} and \textit{receiver} may need to repeat the protocol multiple times to ensure the statistical reliability of the transmitted information. Designing a communication protocol that preserves the information encoded in the diagonal elements of the register state is crucial to mitigating these challenges. 

One approach involves considering a class of interactions that leaves the diagonal elements of the register state unchanged while encoding the message. Resource usage is optimized by ensuring the interaction does not alter the message. Controlled unitary interactions are particularly promising for this purpose. These interactions preserve the register state while executing a state-dependent operation on the system. In general, they can be written as
\begin{equation}\label{eq: control operation}
    U = \sum_x \ketbra{x}{x}\R\otimes U_x,
\end{equation}
where $\{U_x\}$ is a set of unitary operators acting on the system. Applying $U$, Eq.~\eqref{eq: unitary interaction}, to the initial register state $\rho\R = \sum_{x,y} \rho_{x,y} \ketbra{x}{y}$, expressed in a given basis with $p_x = \rho_{x,x}$, and assuming the system is initially prepared in a Gibbs state $\gibbs{\beta}$, we obtain
\begin{equation}\label{eq: rhoRS}
    \trho\RS = \sum_x p_x \ketbra{x}{x}\otimes U_x \rho\Sys U_x^{\dagger}+\trho\off{RS},
\end{equation}
where $\trho\off{RS} = \sum_{y\neq x}\rho_{x,y}\ketbra{x}{y}\otimes U_x \gibbs{\beta} U_y^{\dagger}$, and locally the state of the system is $\trho\Sys = \sum_x p_x\rho_x$. 

As stated in Section~\ref{sec: main theorems}, there are limitations in the rank of the states $\{\rho_x\}$, implying limitations in their distinguishability, as shown in Theorem \ref{theorem: mixture_theorem}. Furthermore, the choice of the set of unitaries $\{ U_x\}$ determines the distinguishability of the $\{\rho_x\}$.
In other words, if each $\rho_x$ belongs to a subspace $\M{H}_x\subset\M{H}\Sys$, spanned by a projector $\Pi_x$, then the states are distinguishable, and the {\it receiver} can have access to $X$. However, as demonstrated in Ref.~\cite{GuryanovaFriisHuber2018}, when considering an initial non-pure system state, constraints exist on the extent of statistical correlations that can be established between the system and the register given optimization over all unitary interactions. 

Eq.~\eqref{eq: cmax} quantifies the maximum statistical correlations that can be established between the system and the register throughout the encoding process. The maximum is achieved by optimizing over the set of all possible unitary operations applied to the joint system.
\begin{eqnarray}\label{eq: cmax}
    \cmax = \max_{\{U\}}\sum_x \tr\left(\ketbra{x}{x}\otimes\Pi_x U (\rho\R\otimes\rho\Sys) U^{\dagger}\right)< 1,
\end{eqnarray}
where $\{\Pi_x\}$ is a given measurement basis of the system and $\{\ketbra{x}{x}\R\}$ is a chosen basis of the register. If the system is prepared in a Gibbs state, as presented in Eq.~\eqref{eq: thermal state}, $\rho\Sys = \sum\exp{(-\beta E_i)} \ketbra{E_i}{E_i}/Z$, with partition function  $ Z = \tr(\exp(-\beta H))$, the quantity $\cmax$ are the sum of the largest $r=d\Sys/d\R$ eigenvalues of $\rho\Sys$, that is,  $\cmax=\sum_{i=0}^{r-1} \exp(-\beta E_i)/Z$.  It reaches $1$ only if the temperature goes to zero, or the energy gaps diverge. The inequality $\cmax < 1$ holds whenever finite thermodynamic resources are considered. In particular, perfect statistical correlations (i.e., $\cmax = 1$) are unattainable under realistic conditions with bounded energy and finite entropy \cite{GuryanovaFriisHuber2018}. Therefore, there is a direct trade-off between the resources used to prepare the system and the statistical correlations established between the register and the system. This trade-off directly affects the \textit{receiver's} ability to distinguish the encoded information. This relationship is stated below and formally proven in Appendix~B.\ref{app_subsubsec: decoding}. 
\begin{theorem}[Optimal ensemble via Controlled Operations]\label{theorem optimal ensemble}
   Consider a quantum communication protocol where the classical information $X$ is encoded in a quantum system via a controlled unitary operation, in Eq.~\eqref{eq: control operation}.~If the \textit{receiver} performs an optimal decoding strategy, the probability of correctly distinguishing the states $\{\rho_x\}$ in the ensemble \textnormal{(}$\psucc$\textnormal{)} is given by
    \begin{equation}\label{eq: psucc is equal to cmax}
        \psucc = \cmax,
    \end{equation}
   where $\cmax$ quantifies the maximum statistical correlations established between the system and the register during the encoding process.
\end{theorem}

\noindent \textit{Proof sketch}: The conditions established in \cite{barnett2009conditions} guarantee that a given projective measurement, generally a positive operator-valued measure (POVM), attains the argument of the optimization problem in Eq.\,\eqref{eq: prob_success}. When a projective measurement is applied to Eq.\,\eqref{eq: cmax} for the optimal unitary $U$, Eq.\,\eqref{eq: psucc is equal to cmax} is satisfied. The optimal protocol details are presented in Appendix\,\,\ref{app_sec: encoding decoding protocol}.

According to Theorem~\ref{theorem optimal ensemble}, the optimal ensemble is generated using controlled unitary operations. This theorem demonstrates that the most effective encoding strategy maximizes both the statistical correlation between the register and the system and the distinguishability of the resulting quantum states. In other words, the protocol that achieves maximal correlation also produces an ensemble that is most accessible to the \textit{receiver} for state discrimination. This result identifies the precise conditions under which the success probability of distinguishing the encoded states reaches its theoretical maximum. From a thermodynamic perspective, the theorem establishes that the extraction of classical information from quantum ensembles is intrinsically limited by the nature of the permissible operations and the constraints of the available resources, including finite thermal states and unitary dynamics. This theorem provides an operational link between optimal information encoding and the physical cost associated with distinguishability.

In order to measure the received data, it is necessary to interact with the system being measured using a measurement apparatus and then analyze the statistics of the pointer. We assume that the \textit{sender} and \textit{receiver} share identical copies of the same thermal baths.~Consequently, the state of the pointer can be prepared in a Gibbs state, $\rho\Poi= \gibbs{\beta}$, following the formalism developed in Refs.~\cite{GuryanovaFriisHuber2018,DebarbaManzanoGuryanovaHuberFriis2019}. 

The transmitted information is encoded in the ensemble denoted by $\Sigma = \{p_x,\, \rho_x\}$, where, in each round, the \textit{receiver} obtains a letter $x$, represented by the state $\rho_x$, with probability $p_x$. After the process ends, the \textit{receiver} acquires a random variable $Y=\{y,\, p(y) \}$, with probability vector elements given by $p(y) = \sum_x p_x p(y|x)$, where $p(y|x) = \tr(\Pi_y \rho_x)$ denotes the probability of measuring $\Pi_y$ given the received state $\rho_x$. 

As demonstrated in Appendix~B.\ref{app_subsubsec: decoding}, the discrepancies between $Y$ and $X$ depend on the statistical correlations established during the encoding process.
\begin{equation}
    || Y - X||_{l_1} \leq 1- \cmax,
\end{equation}
where $|| Y - X||_{l_1} = \sum_{z} |p(z)-p_z |/2$ denotes the $l_1$-distance, and $\cmax$ is defined in Eq.~\eqref{eq: cmax}. The information accessed by the \textit{receiver}, represented by $Y$, in contrast to the information transmitted by the \textit{sender}, represented by $X$, can be quantified by the mutual information between $X$ and $Y$.
\begin{eqnarray}
    I(X\!:\!Y) = H(Y) - H(Y|X),
\end{eqnarray}
wherein $ H(Y) = - \sum_y p(y) \log p(y) $ represents the Shannon entropy of the random variable $ Y $, and $ H(Y|X) = - \sum_{x,y} p_x p(y|x) \log p(y|x) $ denotes the conditional entropy. Using Theorem~\ref{theorem optimal ensemble} and Fano's inequality for the conditional entropy \cite{fano1961transmission}, it follows that the information accessed by the \textit{receiver} satisfies the inequality
\begin{equation}\label{eq: mutual information}
I(X\!:\!Y) \geq H(X) - H_2(\cmax)-(1-\cmax)\log(n-1),
\end{equation}
where $H_2(p) = -p\log p - (1-p)\log(1-p)$ represents the binary entropy. Eq.~\eqref{eq: mutual information} highlights how thermodynamic constraints, encapsulated by the limitation $\cmax<1$, reduce the amount of information accessible to the \textit{receiver}, even with an optimal encoding/decoding protocol.

The two corrective terms in Eq.\,\,\eqref{eq: mutual information} reflect distinct sources of thermodynamic noise. The first term, given by the binary entropy $H_2(\cmax)$, accounts for the statistical mixture inherent in the encoding process. It describes the uncertainty arising from the imperfect distinguishability of the states due to $\cmax<1$. The second term, $(1-\cmax)\log(n-1)$, originates from the correlations between the system and register, which depend on the required effective dimension or the necessity of multiple copies to implement the protocol, as discussed in Ref.~\cite{debarba2024}. These contributions, together, delineate how thermodynamic constraints impact the amount of information retrievable by the \textit{receiver}.

\subsection{Protocol Security and Information Recovery}\label{sec:security}

Theorem\,\,\ref{theorem: mixture_theorem} establishes that, under thermodynamic constraints, the rank condition ensures the indistinguishability of linearly dependent states within the ensemble. This indistinguishability, in turn, guarantees the security of the information transmission protocol, as it prevents unauthorized parties from reliably distinguishing the encoded states. In particular, the impossibility of perfect discrimination between non-orthogonal states limits the ability of a malicious eavesdropper to extract the message, providing a form of thermodynamic protection~\cite{BennettBrassard84,bennett1992}.

The assumption $ \psucc\!=\!1$ implies not only that the states in the ensemble are perfectly distinguishable, but also that an \textit{ideal measurement} exists, i.e., a projective measurement with elements proportional to the states $ \{ \rho_x \} $. However, according to the laws of thermodynamics, such ideal measurements can only be realized with unbounded resources~\cite{GuryanovaFriisHuber2018}, and only when $ \cmax = 1 $. As a result, thermodynamic noise arising from finite-resource constraints limits the amount of accessible information encoded in the ensemble $ \{p_x, \rho_x\} $, as it affects both the preparation and the measurement of the states.

To overcome these limitations, Ref.~\cite{debarba2024} demonstrated that, by using $ N=d\Sys -1 $ copies of the encoded ensemble (or equivalently, by repeating the protocol multiple times), it is possible to construct a decoding strategy in which the \textit{receiver} can post-process the outcomes and reconstruct the original message with high fidelity. This enhanced recovery, however, requires prior knowledge of the thermal state $ \gibbs{\beta} $. In other words, if the \textit{sender} and \textit{receiver} share information about their thermal resources, this shared knowledge enables perfect recovery of the transmitted information.

\section{Thermodynamic Interpretation of Holevo information} \label{sec: holevo and heat}

The Holevo information, denoted by $\chi(\Sigma)$, represents the maximum amount of classical information that can be extracted from the quantum ensemble $\Sigma = \{p_x,\rho_x\}$. The inclusion of this quantity in a thermodynamic balance equation suggests that constraints on energy and entropy are inherent to the storage and retrieval of information. When information is encoded optimally, the reduction in system entropy due to encoding is captured by $\chi(\Sigma)$. However, irreversible processes may occur in realistic scenarios, resulting in additional entropy production  \cite{debarba2024}. 

In this context, the encoding process maps a classical message $X$ into a quantum ensemble characterized by probabilities $\{p_x\}$ and states $\{\rho_x\}$. This process results in a compound state $\tilde{\rho}\RS$, where the register $R$ holds the classical information, and the system $S$ undergoes a transformation governed by unitary operations $U_x$, as presented in Eq.~\eqref{eq: rhoRS}. As discussed in Eq.~\eqref{eq: holevo_inf}, the Holevo information corresponds to a classical-quantum mutual information $\chi(\Sigma)=I(c\!:\!\mathcal{Q})_{\tilde{\rho}\RS^{\Pi}}$ of a local-dephased state $\rho\RS^{\Pi} =  \sum_x p_x\ketbra{x}{x}\otimes\rho_x$. Computing the quantum-classical information $I(c\!:\!\mathcal{Q})_{\tilde{\rho}\RS^{\Pi}} =  S(\sum_x p_x\rho_x) - \sum_x p_x S(\rho_x)$, we obtain
\begin{eqnarray}
I(c\!:\!\mathcal{Q})_{\tilde{\rho}\RS^{\Pi}} = S(\tilde{\rho}\Sys) - S(\rho\Sys) \equiv \Delta S\Sys.\label{DeltaS}
\end{eqnarray}
where in Eq.\,\eqref{DeltaS}, $\Delta S\Sys = S(\tilde{\rho}\Sys) - S(\gibbs{\beta})$ is the system's total entropy variation. On the other hand, as the system and register are initially not correlated, the mutual information of the original $\trho\RS$, $I(R\!:\!S)_{\trho\RS}=S(\tilde{\rho}\R) + S(\tilde{\rho}\Sys) - S(\tilde{\rho}\RS)$, is given by
\begin{eqnarray}
I(\tilde{\rho}\R\!:\!\tilde{\rho}\Sys) = \Delta S\R + \Delta S\Sys = \Delta.\label{SRSS}
\end{eqnarray}
with $\rho\Sys = \gibbs{\beta}$, and $\Delta$ is the total entropy variation. 

By using the Eq.\,\eqref{SRSS} and Theorem 3 of Ref.~\cite{reebandwolf2014}, that reads
\begin{flalign}
I(\tilde{\rho}\R:\tilde{\rho}\Sys) - \Delta S\R = \beta Q - D(\tilde{\rho}\Sys||\rho\Sys).\label{I menos Delta SS}
\end{flalign}

The term $\beta Q$ represents the entropy exchanged with the thermal bath, while $D(\tilde{\rho}\Sys||\rho\Sys)$ quantifies the out-of-equilibrium transformation of the system.
Therefore, by starting with Eq.~\eqref{I menos Delta SS}, and using the relation $\chi(\Sigma) = \Delta S\Sys$ obtained in Eq.~\eqref{DeltaS} in Eq.~\eqref{SRSS} yields
\begin{equation}
\label{resultado_chi_D}
\chi(\Sigma) = \beta Q - D(\tilde{\rho}\Sys||\rho\Sys),
\end{equation}

The relation between relative entropy and free energy is as follows: $\beta\Delta F = D(\tilde{\rho}\Sys||\gibbs{\beta})$, \textit{cf.} \cite{donald1987}
\begin{equation}
\label{resultadoChi}
\chi(\Sigma) = \beta Q - \beta\Delta F.  
\end{equation}

As the system undergoes an out-of-equilibrium transformation, the variation in free energy is always positive. Consequently, the amount of information encoded in the ensemble $\Sigma$ is bounded above by the heat generated in the process, satisfying 
$\chi(\Sigma) \leq \beta Q$. Additionally, ~Eq.~\eqref{resultadoChi} explicitly expresses a protocol that saturates the generalized Holevo-Landauer bound derived in Ref.~\cite{debarba2024}, highlighting the tightness of the bound.~Thus, the control unitary interaction generally represents the optimal scenario for information encoding, as it saturates this bound. This result shows that Holevo information in this quantum encoding-decoding channel is linked to heat exchange and free energy variation. Furthermore, this conclusion aligns with well-established findings in information thermodynamics, such as those involving Maxwell's demon-inspired feedback control mechanisms, where information facilitates work extraction beyond standard thermodynamic constraints \cite{Junior2025}. 

Equation~\eqref{resultadoChi} shows that the necessary heat input has two components: the free energy change $\Delta F$, which represents the minimum energy cost in the absence of information, and an extra term,  $\chi(\Sigma)/\beta$, which quantifies an informational energy cost. This emphasizes the role of information as a thermodynamic resource and demonstrates that accessible information imposes a fundamental bound on the heat input required for any given encoding process.

The relation obtained in Eq.\,\eqref{resultadoChi} offers a thermodynamic interpretation of the Holevo information, showing that it is proportional to the difference between the heat exchanged during the encoding process and the free energy variation of the system. This result connects the amount of accessible information directly to the thermodynamic cost of generating the ensemble under unitary interactions with a thermal resource. It is important to distinguish this approach from the formulation presented by Reeb and Wolf~\cite{reebandwolf2014}, where the primary focus is on the thermodynamic cost of erasure processes. 
While their relation rigorously quantifies the minimal cost of information erasure, our result characterizes the thermodynamic constraints on information encoding, revealing that the heat generated during preparation inherently limits the amount of retrievable information.

If the \textit{receiver} implements a protocol to extract the information encoded in the random variable $Y$, the retrieved information is quantified by the mutual information $I(X\!:\!Y)$. Since the Holevo quantity sets an upper bound on the accessible information in the ensemble $\Sigma$, and the Holevo information itself is constrained by the heat generated during the encoding process. {As shown} in Eq.\,\eqref{resultadoChi}, it follows that the thermodynamic cost of the process limits the decoded information. Specifically,
\begin{equation}
    I(X\!:\!Y) \leq \beta Q,
\end{equation}
where $\beta$ is the inverse temperature of the thermal bath in contact with the measurement apparatus, and $Q$ is the heat exchanged during the encoding process. This inequality illustrates that the decoding fidelity is ultimately limited by the amount of heat generated. Reducing the pointer's temperature or increasing the heat involved in the encoding process alleviates this constraint, enabling the encoding and extraction of more information.

\section{Conclusions and Perspectives}\label{sec: conclusion}
In this work, we have explored the fundamental limitations imposed by thermodynamics on the encoding and decoding of classical information in quantum systems. By integrating principles from quantum mechanics and thermodynamics, we have demonstrated that finite thermodynamic resources constrain the preparation, transmission, and measurement of quantum states used for information encoding. Our results reveal the interplay between quantum coherence, state distinguishability, and thermodynamic costs. This provides a deeper understanding of the inherent trade-offs in quantum communication protocols.

We derived a rank relation for quantum ensembles (Lemma~\ref{lemma: main rank inequality}), which shows that the product of the ranks of the system and register states cannot exceed a scaled version of the ensemble's maximum rank. This result highlights the trade-offs in resources required to increase the number of states or their distinguishability.

As stated in Theorem~\ref{theorem: non-orthogonal ensembles}, we proved that encoding classical information into an ensemble of non-orthogonal pure states is impossible under finite thermodynamic resources. This result directly follows from the third law of thermodynamics, which limits the purity of states that can be produced with finite resources. Theorem~\ref{theorem: mixture_theorem} represents the finite thermodynamic resources version of the mixture Theorem \cite{HJW93}. Theorem~\ref{theorem optimal ensemble} has significant implications for quantum information processing, underscoring the importance of optimal ensemble selection in maximizing information retrieval and establishing a fundamental bound on the distinguishability of quantum states.

In practical scenarios, this result informs the design of optimal encoding schemes for quantum communication and computation. In these scenarios, efficiently extracting classical information from quantum systems is a crucial task. Specifically, it establishes a direct relationship between the Holevo information and the heat produced during the process, given by $\chi(\Sigma) \leq \beta Q$. Our findings emphasize the role of information as a thermodynamic resource and provide a refined second-law-like bound for quantum information processing, wherein encoding more information comes at the cost of increasing heat production.

Our framework establishes an operational thermodynamic limit on information encoding, forming a basis for further progress. A similar study in a classical context, examining minimal entropy production in communication channels and computational actions during encoding and decoding, is detailed in Ref.~\cite{YadavWolpert2024}. Extending these findings to quantum error-correction protocols offers potential for reducing the thermodynamic expenses of large-scale quantum information processing.

Ref.~\cite{Xuereb_2025} investigated the compression of information encoded in an ensemble of pure states while considering thermodynamic limitations. However, our work takes it a step further by asking how the ensemble can be prepared and what limitations thermodynamics imposes. Here we discuss the case in which the message to be sent contains a single letter; for a $l$ letter message, it is necessary to prepare $l$ copies of the ensemble, such that each state of the ensemble $\rho_x\suptiny{0}{0}{(1)}\otimes\cdots\otimes\rho_x\suptiny{0}{0}{(l)}$ is prepared with probability $p_x\suptiny{0}{0}{(1)}\cdots p_x\suptiny{0}{0}{(l)}$, and perform a $l$ letters noisy compression of mixed states \cite{schumacher1997sending,horodecki1998limits}. 
Further research in this area could deepen our understanding of the fundamental limits of quantum information processing and enable practical applications in emerging quantum technologies.

Investigating the thermodynamic efficiency and robustness of autonomous machines could provide valuable insights into designing realistic quantum devices \cite{meier2024autonomous}.~This perspective could establish a direct link between the study of communication-limited encoding and operational models of quantum machines operating under thermodynamic constraints. As well, the possibility of information loss due to interaction with the thermal bath, a topic explored in Ref.~\cite{FaistDupuisOppenheimRenner2015}.

For non-commuting, conserved quantities, the conventional notion of a thermal state is invalid, and appropriate equilibrium states become generalized Gibbs ensembles subject to non-Abelian constraints. These scenarios raise fundamental questions about how information can be encoded, manipulated, and retrieved under generalized thermodynamic laws. Recent theoretical developments \cite{Lostaglio2020, Guryanova2020, Majidy23,Majidy2_23} demonstrate that this line of research offers a promising method to extend thermodynamic communication protocols beyond the traditional equilibrium approach.

Holevo information depends directly on thermal fluctuations, which suggests that in thermodynamically constrained scenarios, Holevo information could proxy indirect temperature estimation.~Ref.~\cite{BaratoSeifert2018} demonstrates the connection between entropy production, fluctuations, and information-theoretic bounds via thermodynamic uncertainty relations. Exploring these connections could inform new protocols for information-based thermometry, especially in contexts where direct access to the bath is limited.

\vspace{0.2cm}
\begin{acknowledgments}
We are grateful to Marcus Huber, Pharnam Bakhshinezad, and Jake Xuereb for fruitful discussions during the early stages of the project and Federico Centrone for the discussions on quantum communication and thermodynamic limitations.~We are deeply indebted to Profa.~Sema Nistia for her exceptional contributions in navigating the intricate obstacles underpinning this work, safeguarding the foundational principles upon which this research rests. The authors acknowledge support from {\"O}AW-JESH-Programme and the Brazilian agencies CNPq (Grant No. 441774/2023-7, 200013/2024-6, and 445150/2024-6) and INCT-IQ through the project (465469/2014-0). ATC acknowledges (RAU $N^\circ\,12-2016$-AY-UNA, Puno - Per\'{u}). ATC thanks Prof. Reinaldo O. Vianna and the Infoquant group for discussions and laboratory infrastructure.
\end{acknowledgments}

\newpage
\hypertarget{sec:appendix}
\onecolumngrid
\appendix

\section*{Appendices}\label{Appendix}
In the Appendices, we present the proof of Lemma~\ref{lemma: main rank inequality} omitted in the main text. We also describe in detail the optimal encoding/decoding protocol discussed in Sec.~\ref {sec: optimal encoding}, which is presented in Appendix~\ref{app_sec: encoding decoding protocol}, as well as the proof of Theorem~\ref{theorem optimal ensemble}.

\section{Proof Lemma~\ref{lemma: main rank inequality}}\label{app_sec: encoding limitations}
\setcounter{lemma}{0}  
\begin{lemma}[Ensemble Rank Relation]\label{app_lemma: main rank inequality}
Consider a register state $\rho\R$, which encodes a random variable $X$ with cardinality $|X| = n$ in its diagonal elements $p_x$. This register unitarily interacts with a system prepared in the state $\rho\Sys$. If the post-interaction state of the system is given by the averaged state $\trho\Sys = \sum_{x=0}^{n-1} p_x \rho_x$, the rank of the states $\rho_x$ will satisfy
\begin{equation}\label{app_eq: main rank inequality_lemma}
  \text{rank}(\rho\Sys) \cdot \text{rank}(\rho\R) \leq n \cdot \text{rank}(\rho_{\text{max}}),
\end{equation}
where $\rho_{\text{max}}$ is the state with the highest rank within the set $\{\rho_x\}_{x=0}^{n-1}$.
\end{lemma}

\begin{proof}
    In general, there exists a unitary interaction $U$ over the system+register such that, for a giving register local basis $\ket{x}\R$ results 
\begin{equation}
    \trho\RS = U(\rho\R\otimes\rho\Sys)U\da = \sum_{x,y=0}^{n-1}q_{x,y}\ketbra{x}{y}\otimes\rho_{x,y},
\end{equation}
such that locally, the system state is given by the ensemble average 
\begin{equation}
    \trho\Sys = \sum_x p_x \rho_x,
\end{equation}
with $q_{x,x}=p_x$ and $\rho_{x,x} = \rho_x$. 
As the interaction is unitary, it preserves the rank, so
\begin{equation}\label{eq: rankSR}
\rank{\trho\RS} = \rank{\rho\R}\rank{\rho\Sys}.
\end{equation}

If a local dephasing operation is applied in the basis $\ket{x}$ to $\trho\RS$, it will ``destroy'' the quantum correlations between the system and the register while keeping the ensemble invariant. This can be represented as
\begin{equation}\label{eq: rhoSR Pi}
 \trho\RS\suptiny{2}{0}{\Pi}   = \sum_{x=0}^{n-1}p_{x}\ketbra{x}{x}\otimes\rho_{x}. 
\end{equation}

The dephasing operation is stochastic and, therefore, it can only increase the rank, thus
\begin{equation}\label{eq: rankSR Pi}
    \rank{\trho\RS} \leq \rank{\trho\RS\suptiny{2}{0}{\Pi}} = \,\rank{\sum_x p_x \rho_x}\leq n\,\rank{\rho\subtiny{0}{0}{max}}, 
\end{equation}
for $\rho\subtiny{0}{0}{max}$ being the state with the biggest rank of the set $\{\rho_x\}_{x=0}^{n-1}$.
Combining Eq.\eqref{eq: rankSR} with \eqref{eq: rankSR Pi}, we have 
\begin{equation}\label{eq: main rank inequality_final}
  \rank{\rho\R}\rank{\rho\Sys}\leq  n\,\rank{\rho\subtiny{0}{0}{max}}.
\end{equation}
\end{proof}

\section{Optimal encoding/decoding protocol}\label{app_sec: encoding decoding protocol}
In this Section, we present the proposed encoding and decoding protocol in detail. The protocol consists of the following steps: 
\begin{enumerate}
    \item[(\ref{app_subsubsec: register_preparation})] \textbf{Encoding the classical variable on $\rho\R$}: The \textit{sender} has free access to thermal states $\gibbs{\beta}$ and prepares the register by encoding the classical random variable $X$ in its diagonal through the application of a Haar-random unitary $ U\R $. The resulting state is given by $\rho\R = U\R \gibbs{\beta} U\R^{\dagger}$, in the form
    \begin{equation}
    \rho\R = \sum_{x=0}^{n-1} p_x \ketbra{x}{x} + \rho^{\text{off}}\R,
    \end{equation}
    Here, the off-diagonal terms are denoted by $ \rho^{\text{off}}\R=\sum_{x\neq y}\bra{x}\rho\R\ket{y}$.
\item[(\ref{app_subsubsec: system_preparation})]\textbf{Preparing the system state $\rho\Sys$}:
    The system is prepared in a thermal state, defined as $\rho\Sys =\gibbs{\beta}$. For a random variable with cardinality $n$, the system's Hilbert space can be decomposed into $n$ subspaces, such that $\M{H}\Sys = \M{H}_0 \oplus \cdots \oplus \M{H}_{n-1}$, where each letter $x$ of $X$ is encoded in a corresponding $d_x-$dimensional subspace given by $\M{H}_x$. This partitioning of the Hilbert space satisfies the condition $d\Sys = \sum_{x=0}^{n-1} d_x$.
\item[(\ref{app_subsubsec: encoding})]\textbf{Transferring the information}
    The transfer and encoding of information from the register to the system occurs through a controlled unitary interaction of the following form
    \begin{equation}
    U = \sum_x \ketbra{x}{x}\R \otimes U_x,
    \end{equation}
    thereby generating the ensemble $ \Sigma = \{p_x, \rho_x\} $ in the system, in a given measurement basis. 
 \item[(\ref{app_subsubsec: decoding})]\textbf{Decoding the information}
    The \textit{receiver} decodes the information by measuring the average state $ \trho_S = \sum_x p_x \rho_x $. This is achieved by coupling $ \rho_x $ to a pointer system, initialized in the thermal state $ \gibbs{\beta} $, through an interaction $V$ that maximizes $\psucc$ as given in Eq. \eqref{eq: prob_success}.
\end{enumerate}

Now, we present the methods related to each step of the protocol. 

\subsection{Preparing the register state}\label{app_subsubsec: register_preparation}
The register is initiated at a Gibbs state $\gibbs{\beta} = \exp(-\beta H)/Z$,  where $Z = \tr(\exp(-\beta H))$, at inverse temperature $\beta$ and Hamiltonian $H = \sum_i E_i \ketbra{E_i}{E_i}$. The information $X$ can be encoded in the diagonal of the register by applying a unitary operation $U\R$ over $\gibbs{\beta}$, as stated in the Remark~\ref{remark orthogonal pure states}, 
\begin{eqnarray}
    \rho\R = U\R \gibbs{\beta} U\R^{\dagger}.
\end{eqnarray}
By considering $u_{i,j}$ being the elements of $U\R$. The state $\rho\R=\sum_{x=0}^{n-1} p_x \ketbra{x}{x}+\rho\off{R}$, where $\rho\off{R}$ are the off-diagonal elements of the register state, can be prepared with the diagonal elements as
\begin{eqnarray}
    p_x = \frac{1}{Z}\sum_j |u_{x,j}|^2\displaystyle e^{-\beta E_j}. 
\end{eqnarray}

\subsection{Preparing the System State}\label{app_subsubsec: system_preparation}

Consider a system prepared in a $d\Sys$-dimensional state $\rho\Sys$. For a random variable with cardinality $n$, the Hilbert space of the system can be decomposed into $n$ subspaces, such that $\M{H}\Sys = \M{H}_1 \oplus \cdots \oplus \M{H}_n$, where each letter $x$ of $X$ is encoded in a corresponding subspace $\M{H}_x$. The partitioning of the Hilbert space depends on $n$ and $d\Sys$. 

In the following subsections, we describe how to coarse-grain the system state when access to a thermal bath with a fixed dimension $ d\Sys $ is available.

\subsubsection{Coarse-Graining the System State for $n\leq d\Sys$}

Consider a random variable $X$ with cardinality $n = d\R$, where $d\R$ is the dimension of the register. The system Hilbert space $\M{H}\Sys$ must be partitioned into $n$ subspaces to encode information. Each subspace has dimension $d_x$ such that $\sum_x d_x = d\Sys$. The system space is spanned by the set of $d_x$-dimensional projectors $\Pi_x = \sum_{l=0}^{d_x-1} \ketbra{\pi_x\suptiny{0}{0}{(l)}}{\pi_x\suptiny{0}{0}{(l)}}$. 

To encode the information of $X$, we must express $\rho\Sys$ in terms of $\Pi_x$. This can be achieved by \textit{coarse-graining} the thermal state of the system, given by $\gibbs{\beta} = \sum_i \exp(-\beta E_i)/Z\ketbra{E_i}{E_i}$, with respect to the relation $i = x d_x + l$, such that $\ket{\pi_x\suptiny{0}{0}{(l)}} = \ket{E_{x d_x + l}}$,  thereby obtaining the blocked state.
\begin{equation}\label{eq: prepared_system}
    \rho\Sys = \sum_{x=0}^{n-1}\sum_{l=0}^{d_x-1} r_x\suptiny{0}{0}{(l)}\ketbra{\pi_x\suptiny{0}{0}{(l)}}{\pi_x\suptiny{0}{0}{(l)}},
\end{equation}
with $r_x\suptiny{0}{0}{(l)} = \exp(-\beta e_x\suptiny{0}{0}{(l)})/Z\Sys$, where $e_x\suptiny{0}{0}{(l)}=E_{x d_x + l}$ are the energy eigenvalues of the system Hamiltonian, given by $H\Sys = \bigoplus_{x=0}^{n-1} H_x = \sum_x \sum_l e_x\suptiny{0}{0}{(l)} \ketbra{\pi_x\suptiny{0}{0}{(l)}}{\pi_x\suptiny{0}{0}{(l)}}$. The partition function of the system is given by $Z\Sys = \sum_x \sum_l \exp(-\beta e_x\suptiny{0}{0}{(l)})$.

\subsubsection{Coarse-Graining the System State for $n > d\Sys$}

Multiple copies of Gibbs states are required when the number of encoded states exceeds the system's dimension $(n > d\Sys)$. Considering the system as composed of $N$ copies of $\gibbs{\beta}$, its state is given by $\rho\Sys = \gibbs{\beta}^{\otimes N}$. Applying a coarse-graining procedure for $|X| \leq d\Sys^N$, we obtain:
\begin{equation}
    \rho\Sys = \sum_{x=0}^{n-1} \sum_{l=0}^{d_x-1} r_x\suptiny{0}{0}{(l)}\ketbra{\pi_x\suptiny{0}{0}{(l)}}{\pi_x\suptiny{0}{0}{(l)}},
\end{equation}
where the coarse-grained basis states are defined as 
\begin{equation}
    \ket{\pi_x\suptiny{0}{0}{(l)}} \equiv \bigotimes_{\mu=1}^N \ket{E_{k\subtiny{0}{0}{\mu}}}.
\end{equation}
The indices satisfy the relation:
\begin{equation}
    xd_x + l = \sum_{\mu=1}^N k_\mu d\Sys^{\mu-1} := f(\vec{k}),
\end{equation}
with $\vec{k} = (k_1, \dots, k_N)$. The variables $x$ and $k$ are determined by
\begin{equation}
  x = \left\lfloor \frac{f(\vec{k})}{d_x} \right\rfloor, \quad k = f(\vec{k}) \mod d_x.
\end{equation}
Here, $\lfloor\frac{a}{b}\rfloor$ denotes the floor function, which returns the greatest integer less than or equal to $\frac{a}{b}$.

\subsection{Information encoding}\label{app_subsubsec: encoding}

Initially, a register is prepared with the random variable $X$ in its diagonal 
\begin{equation}
    \rho\R = \sum_{x=0}^{n-1} p_x\ketbra{x}{x}+ \rho\subtiny{0}{0}{\text{off}},
\end{equation}

\noindent where $\rho_{x,x\pr} = \bra{x}\rho\R\ket{x\pr}$. After the coarse-graining process, the state of the system can be expressed as a composition of each subspace as 
\begin{equation}\label{eq: system blocks}
    \rho\Sys = \sum_{x=0}^{n-1} A_x. 
\end{equation}
The operators $A_x$ are defined as $A_x= \sum_{l=0}^{d_x-1} r_x\suptiny{0}{0}{(l)}\ketbra{\pi_x\suptiny{0}{0}{(l)}}{\pi_x\suptiny{0}{0}{(l)}}$, and represent the components of $\rho\Sys$ in each subspace spanned by $\Pi_x=\sum_{l=0}^{d_x-1} \ketbra{\pi_x\suptiny{0}{0}{(l)}}{\pi_x\suptiny{0}{0}{(l)}}$. It is important to notice that this subspace composition is not unique. Here, we adopt an energy-increasing form, aligning with the thermal nature of the state. The interaction between the system and the register follows a controlled operation, which is described by the following expression 
\begin{equation}
    U = \sum_x \ketbra{x}{x}\R\otimes U_x.
\end{equation}
Acting $U$ over the initial state $\rho\RS= \rho\R\otimes\rho\Sys$, the post-interaction state $\trho\RS$ will be 
\begin{align}\label{eq: rhoSR final}
    \trho\RS &= U\rho\RS U\da  = (\sum_x \ketbra{x}{x}\R\otimes U_x)(\rho\R\otimes\rho\Sys) (\sum_x \ketbra{x}{x}\R\otimes U\da_x),\\
    &= \sum_x p_x \ketbra{x}{x}\otimes U_x \rho\Sys U\da_x + \trho\off{RS},
\end{align}
where $\trho\subtiny{0}{0}{\text{off}} = \sum_{x\neq x\pr} \trho_{x,x\pr} \tr(U_x \rho\Sys U_{x\pr}\da)\ketbra{x}{x\pr}$, and $\trho_{x,x\pr} = \bra{x}\trho\R\ket{x\pr}$, with $\rho_{x,x}=p_x$. Observe that the diagonal of the register state remains unchanged.  Therefore, a register prepared once (via a Haar-random $U\R$) can be reused for arbitrarily many subsequent controlled-$U$ interaction/measurement rounds, since these operations leave its diagonal invariant and do not require re-preparation of the register. {This implies that the system state becomes}
\begin{equation}\label{eq: rhoS final}
    \trho\Sys = \sum_x p_x U_x \rho\Sys U_x\da = \sum_x p_x\rho_x.
\end{equation}
Now, considering $\rho\Sys$ as presented in Eq.~\eqref{eq: system blocks}, we must determine the effect of the unitaries $U_x$ on the components of $A_x$. To proceed, we introduce 
\begin{equation}
    U_x A_y U_x\da = \sum_{l=0}^{d_x-1} r_y\suptiny{0}{0}{(l)}U_x\ketbra{\pi_y\suptiny{0}{0}{(l)}}{\pi_y\suptiny{0}{0}{(l)}}U_x\da,
\end{equation}
and, therefore, the states of the ensemble are determined by the action of $U_x$ on the system Hamiltonian eigenbasis $\{\ket{\pi_y\suptiny{0}{0}{(l)}}\}$. 

In quantum thermodynamics, coherence is a valuable resource, and its manipulation requires an expenditure of additional energy \cite{lostaglio15,misra16,lostaglio19,Francica20,Gour22}. Consequently, only partial swaps are typically employed in such processes. Thus, the unitaries $U_x$ should act solely on the diagonal elements of the system state as
\begin{equation}
    U_x\ket{\pi_y\suptiny{0}{0}{(l)}} = \ket{\pi_{(y+x)}\suptiny{0}{0}{(l)}},
\end{equation}
where $(y+x):=(y+x)\mod n$. Later, in Theorem \ref{app_theorem optimal ensemble}, we will demonstrate that this choice of transformation yields the optimal ensemble for encoding the random variable $X$. For the sake of completeness in notation, the final state of the system can be expressed as
\begin{equation}
    \trho\Sys = \sum_x p_x \rho_x, 
\end{equation}
with $\rho_x = \sum_y U_x A_y U_x\da$. Since summing over $y$ in $U_x A_y U_x\da$ intermixes all subspaces, the states $\rho_x$ become indistinguishable. Consequently, full information cannot be extracted from the ensemble $\Sigma = \{p_x,\rho_x\}$. Notice that the indistinguishability of the $\rho_x$ is related to the overlap of the subspaces with orthogonal support, spanned by $\Pi_{z\neq x}$, 
\begin{align}\label{eq: overlap Pi and sigma}
    \tr(\Pi_z \rho_x) &= \sum_y \tr(\Pi_z U_x A_y U_x\da) = \sum_y \sum_l r_y\suptiny{0}{0}{(l)} \bra{\pi_{(y+x)}\suptiny{0}{0}{(l)}} \Pi_z \ket{\pi_{(y+x)}\suptiny{0}{0}{(l)}}= \sum_l r_{(z-x)}\suptiny{0}{0}{(l)}, 
\end{align}
where, in the second equality, we used the following property  
\begin{equation}\label{eq: overlap Pi with energy eigenvalues}
    \bra{\pi_{(y+x)}\suptiny{0}{0}{(l)}} \Pi_z \ket{\pi_{(y+x)}\suptiny{0}{0}{(l\pr)}} = \sum_{l\pr}  |\braket{\pi_{(y+x)}\suptiny{0}{0}{(l)}}{\pi_{z}^{(l\pr)}}|^2 = \sum_{l\pr} \delta_{(y+x),z}\delta_{l,l\pr}.
\end{equation}
For $x=z$, Eq.~\eqref{eq: overlap Pi and sigma} quantifies how distinguishable the states of the ensemble are: $\tr(\Pi_x \rho_x) = \sum_l r_{(0)}\suptiny{0}{0}{(l)}$. This quantity, as proved in Ref.~\cite{GuryanovaFriisHuber2018}, expresses the maximum amount of statistical correlations between the system and register, $\cmax = \sum_l r_{(0)}\suptiny{0}{0}{(l)}$. 
As stated in the following theorem, the choice of the unitaries $U_x$ and the projectors $\Pi_x$ is related to the maximization of the statistical correlations between the system and register, as well as the optimization of the distinguishability among the states of the ensemble. 

\setcounter{theorem}{2}
\begin{theorem}[Optimal ensemble via Controlled Operations]\label{app_theorem optimal ensemble}
  Consider a quantum communication protocol where the classical information $X$ is encoded in a quantum system via a controlled unitary operation of the form 
   \begin{equation*}
       U = \sum_x \ketbra{x}{x}\R \otimes U_x.
   \end{equation*} 
   If the \textit{receiver} performs an optimal decoding strategy, the probability of correctly distinguishing the states $\{\rho_x\}$ in the ensemble ($\psucc$) is given by
    \begin{equation}\label{app_eq: psucc is equal to cmax}
        \psucc = \cmax,
    \end{equation}
   where $\cmax$ quantifies the maximum statistical correlations established between the system and the register during the encoding process.
\end{theorem}
\begin{proof}
Given the state $\trho\RS=U\rho\RS U^{\dagger}$ as computed in Eq.~\eqref{eq: rhoSR final} and the relation obtained in Eq.~\eqref{eq: overlap Pi and sigma}, for a projective measurement with elements $\Pi_x$, the statistical correlation between the system and register is 
\begin{equation}
 C(\trho\RS) = \sum_x \tr(\ketbra{x}{x}\otimes\Pi_x\trho\RS)=\sum_xp_x\sum_l r_{(0)}\suptiny{0}{0}{(l)}=\sum_l r_{(0)}\suptiny{0}{0}{(l)},
\end{equation}
which is the maximum statistical correlation ($\cmax = \max_{\{U\}}C(\trho\RS)$) that can be created between a classical register and a thermal state, as shown in Ref.~\cite{GuryanovaFriisHuber2018}. 

 For an ensemble $\Sigma = \{p_x,\rho_x\}_{x=0}^{n-1}$, the probability of success in distinguishing the states of the ensemble is defined as 
\begin{equation}
    P_{suc} = \max_{\substack{0\leq P_x\leq \id, \\ \sum_x P_x =\id}}\left(\sum_x p_x\tr(\rho_xP_x)\right).
\end{equation}
This optimization problem cannot, in general, be analytically determined. However, if a POVM, with elements $ P_x $, satisfies the following properties, it is the unique argument of the probability of success \cite{barnett2009conditions} 
\begin{align}
\label{eq: Psucc condition 1}    &P_x(p_x\rho_x - p_y\rho_y)P_y =0, \, \forall\, x,y, \\
\label{eq: Psucc condition 2}    &\sum_x  p_x\rho_xP_x - p_y\rho_y \geq 0, \, \forall\, y.
\end{align}
Considering the projective measurement, with elements $P_x$ and the set of states $\rho_y$ satisfying Eq.~\eqref{eq: overlap Pi and sigma}, we obtain
\begin{align}
    \Pi_x \rho_x &= \sum_y \sum_l r_{y}\suptiny{0}{0}{(l)}\delta_{y+x,x}\ketbra{\pi_{(x)}}{\pi_{(y+x)}},\\
    &= \sum_l r_{0}\suptiny{0}{0}{(l)}\ketbra{\pi_{(x)}}{\pi_{(x)}}. 
\end{align}
Therefore, condition \eqref{eq: Psucc condition 1} satisfies, 
\begin{align}
    \Pi_x(p_x\rho_x - p_y\rho_y)\Pi_y &= p_x \sum_l r_{0}\suptiny{0}{0}{(l)}\ketbra{\pi_{(x)}}{\pi_{(x)}}\Pi_y - p_y \sum_l r_{0}\suptiny{0}{0}{(l)}\Pi_x\ketbra{\pi_{(y)}}{\pi_{(y)}},\\
    &= \left( p_x \sum_l r_{0}\suptiny{0}{0}{(l)}\ketbra{\pi_{(x)}}{\pi_{(y)}} - p_y \sum_l r_{0}\suptiny{0}{0}{(l)}\ketbra{\pi_{(x)\suptiny{0}{0}{(l)}}\suptiny{0}{0}{(l)}}{\pi_{(y)}} \right)\delta_{x,y},\\
    &=0,\quad \forall\, x,y.
\end{align}
As $\sum_l  r_{0}\suptiny{0}{0}{(l)}\geq \sum_l  r_{x}\suptiny{0}{0}{(l)}\, \forall\, x$, we can compute
\begin{equation}
    \sum_xp_x\rho_x\Pi_x = \sum_x p_x \sum_l  r_{0}\suptiny{0}{0}{(l)} \ketbra{\pi_{(x)}}{\pi_{(x)}} \geq \sum_x p_x \sum_l  r_{x}\suptiny{0}{0}{(l)} \ketbra{\pi_{(x)}}{\pi_{(x)}} \geq p_x \sum_l  r_{x}\suptiny{0}{0}{(l)} \ketbra{\pi_{(x)}\suptiny{0}{0}{(l)}}{\pi_{(x)}\suptiny{0}{0}{(l)}},
\end{equation}
which satisfies condition \eqref{eq: Psucc condition 2} and complete the proof as $\sum_x p_x\tr(\rho_x\Pi_x) = \cmax$. 
\end{proof}

\subsection{Information decoding}\label{app_subsubsec: decoding}
For measuring the received data, it is necessary to interact with the system being measured with the measurement apparatus, and then read the statistics of the pointer. Initially, the \textit{receiver} prepares the measurement apparatus in a state denoted by $\rho\Poi$. We assume that the \textit{sender} and \textit{receiver} share perfect copies of the same thermal baths. Consequently, to comply with thermodynamic constraints, we assume that the state of the pointer is prepared in a coarse-grained Gibbs state, as presented in Eq.~\eqref{eq: prepared_system}
\begin{equation}
    \rho\Poi = \sum_{x=0}^{n-1}\sum_{l=0}^{d_x-1} r_x\suptiny{0}{0}{(l)}\ketbra{\pi_x\suptiny{0}{0}{(l)}}{\pi_x\suptiny{0}{0}{(l)}},
\end{equation}
with $r_x\suptiny{0}{0}{(l)} = \exp{(-\beta e_x\suptiny{0}{0}{(l)})}/Z\Poi$, $e_x\suptiny{0}{0}{(l)}$ are the energy eigenvalues of the system Hamiltonian $H\Poi = \bigoplus_{x=0}^{n-1} H_x =\sum_x(\sum_l e_x\suptiny{0}{0}{(l)} \ketbra{\pi_x\suptiny{0}{0}{(l)}}{\pi_x\suptiny{0}{0}{(l)}})$, and $Z\Poi = \sum_x\sum_l \exp{(-\beta e_x\suptiny{0}{0}{(l)})}$ is the partition function of the pointer state. 
In this work, we follow the non-ideal measurement process outlined in Refs.~\cite{GuryanovaFriisHuber2018,DebarbaManzanoGuryanovaHuberFriis2019}. In this framework, the statistics of the state being measured can be replicated in the pointer if the system and pointer are maximally correlated, as quantified by $\cmax$. This replication occurs through an interaction governed by an {\it unbiased} unitary, resulting in a loss of the original statistics of the measured system.

The information transmitted will be encoded in the ensemble $\Sigma = \{ p_x,\, \rho_x\}$, where in each round the \textit{receiver} will get a letter $x$, represented by the state $\rho_x$, with probability $p_x$. As the information was encoded on the subspaces spanned by $\Pi_x = \sum_l\ketbra{\pi_{x}\suptiny{0}{0}{(l)}}{\pi_{x}\suptiny{0}{0}{(l)}}$, a good strategy for the \textit{receiver} is to measure in the same basis, as it optimizes the probability of success in discriminating the received states, as stated in Theorem~\ref{app_theorem optimal ensemble}. For a given received state $\rho_x$ the probability of obtaining a result $\Pi_y$, as computed in Eq.~\eqref{eq: overlap Pi and sigma}, will be 
\begin{equation}\label{eq: prop y given x}
    p(y|x) = \tr(\Pi_y\rho_x) = \sum_l r_{(y-x)}\suptiny{0}{0}{(l)}, 
\end{equation}
where $r_{(y-x)}\suptiny{0}{0}{(l)}$ given in Eq.~\eqref{eq: prepared_system} represents the diagonal elements of the transformed state of the pointer. Since the optimal protocol maximizes both the distinguishability of the states and the statistical correlations created between the system and the pointer, the probability of measuring $\Pi_x$ given $\rho_x$ is also maximized.
\begin{equation}\label{eq: prob x given x}
     p(x|x) = \sum_l r_{0}\suptiny{0}{0}{(l)} = \cmax = \psucc,\quad \forall\, x=\{0,\ldots,n-1\}. 
\end{equation}
Therefore, the information extracted by the \textit{receiver} ultimately depends solely on preparing the pointer state. Specifically, the larger the populations in the subspace $\Pi_0$, the closer the conditional probability $p(x|x)$ will approach one. On average, the information extracted by the \textit{receiver} is characterized by the probability vector, whose elements are given by 
\begin{eqnarray}
    p(y) = \sum_x p_x p(y|x),
\end{eqnarray}
where $ p_x $ are the probabilities of the original information and $ p(y|x) $ are given by Eq.~\eqref{eq: prop y given x}. The random variable obtained by the \textit{receiver}, $ Y = \{y,\,p(y)\} $, becomes closer to the original message $ X = \{x,\,p_x\} $ as the $ l_1 $-norm is upper bounded by the probability of error, which is given by $ 1 - \psucc = 1 - \cmax $.
\begin{align}
    ||Y-X||_{l_1} 
    &= \frac{1}{2}\sum_y|p(y) - p_y| = \frac{1}{2}\sum_y|\sum_x p_x p(y|x) - p_y|,\\
    &= \frac{1}{2}\sum_y|\sum_{x\neq y} p_x p(y|x)+ p_yp(y|y) - p_y| = \frac{1}{2}\sum_y|\sum_{x\neq y} p_x p(y|x)+ p_y(1-\cmax)|,\\
    &\leq \frac{1}{2}\sum_y\sum_{x\neq y} p_x p(y|x) + \frac{1}{2}\sum_yp_y(1-\cmax),\\ \label{eq: distance Y-X}
    & = 1-\cmax = 1-\psucc,
\end{align}
where we used the triangle inequality and $\sum_y\sum_{x\neq y} p_x p(y|x) = \sum_x\sum_{y\neq x} p_x p(y|x)= 1- \cmax$. 

The information accessed by the \textit{receiver}, in contrast with the information transmitted by the \textit{sender}, can be quantified by the mutual information between $X$ and $Y$
\begin{eqnarray}
    I(X\!:\!Y) = H(Y) - H(Y|X),
\end{eqnarray}
where $ H(Y) = - \sum_y p(y) \log p(y) $ is the Shannon information of the random variable $ Y $, and $ H(Y|X) = - \sum_{x,y} p_x p(y|x) \log p(y|x) $ is the conditional entropy. Applying inequality \eqref{eq: distance Y-X} within Fano's inequality $ H(Y|X) \leq H_2(\psucc) + (1-\psucc) \log(n-1) $ \cite{fano1961transmission}, and using Theorem~\ref{app_theorem optimal ensemble}, one obtains
\begin{eqnarray}\label{app_eq: mutual information}
    I(X\!:\!Y) \geq H(X) - H_2(\cmax)-(1-\cmax)\log(n-1),
\end{eqnarray}
where $H_2(p) = -p\log p - (1-p)\log(1-p)$ is the binary entropy and $H(X)\leq H(Y)$ as $p(y) = \sum_x p_x p(y|x)$ and the Shannon entropy is a Schur-concave function.

We observe in Eq.~\eqref{app_eq: mutual information} that the original information, quantified by $H(X)$, is diminished by the imperfections in the correlations created between the register and the system during the encoding process. These imperfections result in limitations on the amount of information that the \textit{receiver} can access. During the encoding process, an amount of information $\chi(\Sigma)$ is transferred to the ensemble $\Sigma = \{p_x, \rho_x\}$, and therefore, during the decoding process, the \textit{receiver} can recover an amount $I(X:Y)$ of the original information $H(X)$. Consequently, we can compute these information limitations by using the following chain inequality
\begin{equation}
    H(X)\geq \chi(\Sigma)\geq I(X\!:\!Y)\geq H(X) - H_2(\cmax)-(1-\cmax)\log(n-1).
\end{equation}




\begin{thebibliography}{69}%
\makeatletter
\providecommand \@ifxundefined [1]{%
 \@ifx{#1\undefined}
}%
\providecommand \@ifnum [1]{%
 \ifnum #1\expandafter \@firstoftwo
 \else \expandafter \@secondoftwo
 \fi
}%
\providecommand \@ifx [1]{%
 \ifx #1\expandafter \@firstoftwo
 \else \expandafter \@secondoftwo
 \fi
}%
\providecommand \natexlab [1]{#1}%
\providecommand \enquote  [1]{``#1''}%
\providecommand \bibnamefont  [1]{#1}%
\providecommand \bibfnamefont [1]{#1}%
\providecommand \citenamefont [1]{#1}%
\providecommand \href@noop [0]{\@secondoftwo}%
\providecommand \href [0]{\begingroup \@sanitize@url \@href}%
\providecommand \@href[1]{\@@startlink{#1}\@@href}%
\providecommand \@@href[1]{\endgroup#1\@@endlink}%
\providecommand \@sanitize@url [0]{\catcode `\\12\catcode `\$12\catcode `\&12\catcode `\#12\catcode `\^12\catcode `\_12\catcode `\%12\relax}%
\providecommand \@@startlink[1]{}%
\providecommand \@@endlink[0]{}%
\providecommand \url  [0]{\begingroup\@sanitize@url \@url }%
\providecommand \@url [1]{\endgroup\@href {#1}{\urlprefix }}%
\providecommand \urlprefix  [0]{URL }%
\providecommand \Eprint [0]{\href }%
\providecommand \doibase [0]{https://doi.org/}%
\providecommand \selectlanguage [0]{\@gobble}%
\providecommand \bibinfo  [0]{\@secondoftwo}%
\providecommand \bibfield  [0]{\@secondoftwo}%
\providecommand \translation [1]{[#1]}%
\providecommand \BibitemOpen [0]{}%
\providecommand \bibitemStop [0]{}%
\providecommand \bibitemNoStop [0]{.\EOS\space}%
\providecommand \EOS [0]{\spacefactor3000\relax}%
\providecommand \BibitemShut  [1]{\csname bibitem#1\endcsname}%
\let\auto@bib@innerbib\@empty
\bibitem [{\citenamefont {Landauer}(1961)}]{landauer1961}%
  \BibitemOpen
  \bibfield  {author} {\bibinfo {author} {\bibfnamefont {R.}~\bibnamefont {Landauer}},\ }\bibfield  {title} {\bibinfo {title} {{Irreversibility and Heat Generation in the Computing Process}},\ }\href {https://doi.org/10.1147/rd.53.0183} {\bibfield  {journal} {\bibinfo  {journal} {IBM J. Res. Dev.}\ }\textbf {\bibinfo {volume} {5}},\ \bibinfo {pages} {183} (\bibinfo {year} {1961})}\BibitemShut {NoStop}%
\bibitem [{\citenamefont {Anderson}(2017)}]{Anderson2017}%
  \BibitemOpen
  \bibfield  {author} {\bibinfo {author} {\bibfnamefont {N.~G.}\ \bibnamefont {Anderson}},\ }\bibfield  {title} {\bibinfo {title} {Information as a physical quantity},\ }\href {https://api.semanticscholar.org/CorpusID:31589812} {\bibfield  {journal} {\bibinfo  {journal} {Inf. Sci.}\ }\textbf {\bibinfo {volume} {415}},\ \bibinfo {pages} {397} (\bibinfo {year} {2017})}\BibitemShut {NoStop}%
\bibitem [{\citenamefont {Bennett}\ and\ \citenamefont {Brassard}(2014)}]{BennettBrassard84}%
  \BibitemOpen
  \bibfield  {author} {\bibinfo {author} {\bibfnamefont {C.~H.}\ \bibnamefont {Bennett}}\ and\ \bibinfo {author} {\bibfnamefont {G.}~\bibnamefont {Brassard}},\ }\bibfield  {title} {\bibinfo {title} {Quantum cryptography: Public key distribution and coin tossing},\ }\href {https://doi.org/10.1016/j.tcs.2014.05.025} {\bibfield  {journal} {\bibinfo  {journal} {Theoretical Computer Science}\ }\textbf {\bibinfo {volume} {560}},\ \bibinfo {pages} {7} (\bibinfo {year} {2014})},\ \bibinfo {note} {originally presented at the International Conference on Computers, Systems \& Signal Processing, Bangalore, India, December 1984}\BibitemShut {NoStop}%
\bibitem [{\citenamefont {Ekert}(1991)}]{Ekert91}%
  \BibitemOpen
  \bibfield  {author} {\bibinfo {author} {\bibfnamefont {A.~K.}\ \bibnamefont {Ekert}},\ }\bibfield  {title} {\bibinfo {title} {Quantum cryptography based on bell's theorem},\ }\href {https://doi.org/10.1103/PhysRevLett.67.661} {\bibfield  {journal} {\bibinfo  {journal} {Phys. Rev. Lett.}\ }\textbf {\bibinfo {volume} {67}},\ \bibinfo {pages} {661} (\bibinfo {year} {1991})}\BibitemShut {NoStop}%
\bibitem [{\citenamefont {Bennett}(1992)}]{bennett1992}%
  \BibitemOpen
  \bibfield  {author} {\bibinfo {author} {\bibfnamefont {C.~H.}\ \bibnamefont {Bennett}},\ }\bibfield  {title} {\bibinfo {title} {Quantum cryptography using any two nonorthogonal states},\ }\href {https://doi.org/https://journals.aps.org/prl/abstract/10.1103/PhysRevLett.68.3121} {\bibfield  {journal} {\bibinfo  {journal} {Physical review letters}\ }\textbf {\bibinfo {volume} {68}},\ \bibinfo {pages} {3121} (\bibinfo {year} {1992})}\BibitemShut {NoStop}%
\bibitem [{\citenamefont {Gisin}\ and\ \citenamefont {Massar}(1997)}]{GisinMassar1997}%
  \BibitemOpen
  \bibfield  {author} {\bibinfo {author} {\bibfnamefont {N.}~\bibnamefont {Gisin}}\ and\ \bibinfo {author} {\bibfnamefont {S.}~\bibnamefont {Massar}},\ }\bibfield  {title} {\bibinfo {title} {{Optimal Quantum Cloning Machines}},\ }\href {https://doi.org/10.1103/PhysRevLett.79.2153} {\bibfield  {journal} {\bibinfo  {journal} {Phys. Rev. Lett.}\ }\textbf {\bibinfo {volume} {79}},\ \bibinfo {pages} {2153} (\bibinfo {year} {1997})},\ \Eprint {https://arxiv.org/abs/quant-ph/9705046} {arXiv:quant-ph/9705046} \BibitemShut {NoStop}%
\bibitem [{\citenamefont {Bennett}\ \emph {et~al.}(1993)\citenamefont {Bennett}, \citenamefont {Brassard}, \citenamefont {Cr\'{e}peau}, \citenamefont {Jozsa}, \citenamefont {Peres},\ and\ \citenamefont {Wootters}}]{BennettEtAl93}%
  \BibitemOpen
  \bibfield  {author} {\bibinfo {author} {\bibfnamefont {C.~H.}\ \bibnamefont {Bennett}}, \bibinfo {author} {\bibfnamefont {G.}~\bibnamefont {Brassard}}, \bibinfo {author} {\bibfnamefont {C.}~\bibnamefont {Cr\'{e}peau}}, \bibinfo {author} {\bibfnamefont {R.}~\bibnamefont {Jozsa}}, \bibinfo {author} {\bibfnamefont {A.}~\bibnamefont {Peres}},\ and\ \bibinfo {author} {\bibfnamefont {W.~K.}\ \bibnamefont {Wootters}},\ }\bibfield  {title} {\bibinfo {title} {Teleporting an unknown quantum state via dual classical and einstein-podolsky-rosen channels},\ }\href {https://doi.org/10.1103/PhysRevLett.70.1895} {\bibfield  {journal} {\bibinfo  {journal} {Phys. Rev. Lett.}\ }\textbf {\bibinfo {volume} {70}},\ \bibinfo {pages} {1895} (\bibinfo {year} {1993})}\BibitemShut {NoStop}%
\bibitem [{\citenamefont {Zurek}(1991)}]{zurek1991decoherence}%
  \BibitemOpen
  \bibfield  {author} {\bibinfo {author} {\bibfnamefont {W.~H.}\ \bibnamefont {Zurek}},\ }\bibfield  {title} {\bibinfo {title} {Decoherence and the transition from quantum to classical},\ }\href {https://doi.org/10.1063/1.881293} {\bibfield  {journal} {\bibinfo  {journal} {Physics today}\ }\textbf {\bibinfo {volume} {44}},\ \bibinfo {pages} {36} (\bibinfo {year} {1991})}\BibitemShut {NoStop}%
\bibitem [{\citenamefont {Xuereb}\ \emph {et~al.}(2023)\citenamefont {Xuereb}, \citenamefont {Erker}, \citenamefont {Meier}, \citenamefont {Mitchison},\ and\ \citenamefont {Huber}}]{xuerebprl}%
  \BibitemOpen
  \bibfield  {author} {\bibinfo {author} {\bibfnamefont {J.}~\bibnamefont {Xuereb}}, \bibinfo {author} {\bibfnamefont {P.}~\bibnamefont {Erker}}, \bibinfo {author} {\bibfnamefont {F.}~\bibnamefont {Meier}}, \bibinfo {author} {\bibfnamefont {M.~T.}\ \bibnamefont {Mitchison}},\ and\ \bibinfo {author} {\bibfnamefont {M.}~\bibnamefont {Huber}},\ }\bibfield  {title} {\bibinfo {title} {Impact of imperfect timekeeping on quantum control},\ }\href {https://arxiv.org/abs/2301.10767} {\bibfield  {journal} {\bibinfo  {journal} {Physical Review Letters}\ }\textbf {\bibinfo {volume} {131}},\ \bibinfo {pages} {160204} (\bibinfo {year} {2023})}\BibitemShut {NoStop}%
\bibitem [{\citenamefont {Masanes}\ and\ \citenamefont {Oppenheim}(2017)}]{masanes2017general}%
  \BibitemOpen
  \bibfield  {author} {\bibinfo {author} {\bibfnamefont {L.}~\bibnamefont {Masanes}}\ and\ \bibinfo {author} {\bibfnamefont {J.}~\bibnamefont {Oppenheim}},\ }\bibfield  {title} {\bibinfo {title} {A general derivation and quantification of the third law of thermodynamics},\ }\href {https://doi.org/10.1038/ncomms14538} {\bibfield  {journal} {\bibinfo  {journal} {Nature communications}\ }\textbf {\bibinfo {volume} {8}},\ \bibinfo {pages} {14538} (\bibinfo {year} {2017})}\BibitemShut {NoStop}%
\bibitem [{\citenamefont {Wilming}\ and\ \citenamefont {Gallego}(2017)}]{wilming2017third}%
  \BibitemOpen
  \bibfield  {author} {\bibinfo {author} {\bibfnamefont {H.}~\bibnamefont {Wilming}}\ and\ \bibinfo {author} {\bibfnamefont {R.}~\bibnamefont {Gallego}},\ }\bibfield  {title} {\bibinfo {title} {Third law of thermodynamics as a single inequality},\ }\href {https://journals.aps.org/prx/abstract/10.1103/PhysRevX.7.041033} {\bibfield  {journal} {\bibinfo  {journal} {Physical Review X}\ }\textbf {\bibinfo {volume} {7}},\ \bibinfo {pages} {041033} (\bibinfo {year} {2017})}\BibitemShut {NoStop}%
\bibitem [{\citenamefont {Taranto}\ \emph {et~al.}(2023)\citenamefont {Taranto}, \citenamefont {Bakhshinezhad}, \citenamefont {Bluhm}, \citenamefont {Silva}, \citenamefont {Friis}, \citenamefont {Lock}, \citenamefont {Vitagliano}, \citenamefont {Binder}, \citenamefont {Debarba}, \citenamefont {Schwarzhans}, \citenamefont {Clivaz},\ and\ \citenamefont {Huber}}]{Taranto2023}%
  \BibitemOpen
  \bibfield  {author} {\bibinfo {author} {\bibfnamefont {P.}~\bibnamefont {Taranto}}, \bibinfo {author} {\bibfnamefont {F.}~\bibnamefont {Bakhshinezhad}}, \bibinfo {author} {\bibfnamefont {A.}~\bibnamefont {Bluhm}}, \bibinfo {author} {\bibfnamefont {R.}~\bibnamefont {Silva}}, \bibinfo {author} {\bibfnamefont {N.}~\bibnamefont {Friis}}, \bibinfo {author} {\bibfnamefont {M.~P.~E.}\ \bibnamefont {Lock}}, \bibinfo {author} {\bibfnamefont {G.}~\bibnamefont {Vitagliano}}, \bibinfo {author} {\bibfnamefont {F.~C.}\ \bibnamefont {Binder}}, \bibinfo {author} {\bibfnamefont {T.}~\bibnamefont {Debarba}}, \bibinfo {author} {\bibfnamefont {E.}~\bibnamefont {Schwarzhans}}, \bibinfo {author} {\bibfnamefont {F.}~\bibnamefont {Clivaz}},\ and\ \bibinfo {author} {\bibfnamefont {M.}~\bibnamefont {Huber}},\ }\bibfield  {title} {\bibinfo {title} {{Landauer Versus Nernst: What is the True Cost of Cooling a Quantum System?}},\ }\href {https://doi.org/10.1103/PRXQuantum.4.010332} {\bibfield  {journal} {\bibinfo  {journal} {PRX
  Quantum}\ }\textbf {\bibinfo {volume} {4}},\ \bibinfo {pages} {010332} (\bibinfo {year} {2023})},\ \Eprint {https://arxiv.org/abs/2106.05151} {arXiv:2106.05151} \BibitemShut {NoStop}%
\bibitem [{\citenamefont {Taranto}\ \emph {et~al.}(2025)\citenamefont {Taranto}, \citenamefont {Lipka-Bartosik}, \citenamefont {Rodr{\'\i}guez-Briones}, \citenamefont {Perarnau-Llobet}, \citenamefont {Friis}, \citenamefont {Huber},\ and\ \citenamefont {Bakhshinezhad}}]{taranto2025efficiently}%
  \BibitemOpen
  \bibfield  {author} {\bibinfo {author} {\bibfnamefont {P.}~\bibnamefont {Taranto}}, \bibinfo {author} {\bibfnamefont {P.}~\bibnamefont {Lipka-Bartosik}}, \bibinfo {author} {\bibfnamefont {N.~A.}\ \bibnamefont {Rodr{\'\i}guez-Briones}}, \bibinfo {author} {\bibfnamefont {M.}~\bibnamefont {Perarnau-Llobet}}, \bibinfo {author} {\bibfnamefont {N.}~\bibnamefont {Friis}}, \bibinfo {author} {\bibfnamefont {M.}~\bibnamefont {Huber}},\ and\ \bibinfo {author} {\bibfnamefont {P.}~\bibnamefont {Bakhshinezhad}},\ }\bibfield  {title} {\bibinfo {title} {Efficiently cooling quantum systems with finite resources: Insights from thermodynamic geometry},\ }\href {https://arxiv.org/abs/2404.06649} {\bibfield  {journal} {\bibinfo  {journal} {Physical Review Letters}\ }\textbf {\bibinfo {volume} {134}},\ \bibinfo {pages} {070401} (\bibinfo {year} {2025})}\BibitemShut {NoStop}%
\bibitem [{\citenamefont {Lostaglio}\ \emph {et~al.}(2015)\citenamefont {Lostaglio}, \citenamefont {Jennings},\ and\ \citenamefont {Rudolph}}]{lostaglio15}%
  \BibitemOpen
  \bibfield  {author} {\bibinfo {author} {\bibfnamefont {M.}~\bibnamefont {Lostaglio}}, \bibinfo {author} {\bibfnamefont {D.}~\bibnamefont {Jennings}},\ and\ \bibinfo {author} {\bibfnamefont {T.}~\bibnamefont {Rudolph}},\ }\bibfield  {title} {\bibinfo {title} {Description of quantum coherence in thermodynamic processes requires constraints beyond free energy},\ }\href {https://doi.org/10.1038/ncomms7383} {\bibfield  {journal} {\bibinfo  {journal} {Nature communications}\ }\textbf {\bibinfo {volume} {6}},\ \bibinfo {pages} {6383} (\bibinfo {year} {2015})}\BibitemShut {NoStop}%
\bibitem [{\citenamefont {Lostaglio}(2019{\natexlab{a}})}]{lostaglio19}%
  \BibitemOpen
  \bibfield  {author} {\bibinfo {author} {\bibfnamefont {M.}~\bibnamefont {Lostaglio}},\ }\bibfield  {title} {\bibinfo {title} {An introductory review of the resource theory approach to thermodynamics},\ }\href {https://doi.org/10.1088/1361-6633/ab46e5} {\bibfield  {journal} {\bibinfo  {journal} {Reports on Progress in Physics}\ }\textbf {\bibinfo {volume} {82}},\ \bibinfo {pages} {114001} (\bibinfo {year} {2019}{\natexlab{a}})}\BibitemShut {NoStop}%
\bibitem [{\citenamefont {Gour}(2022)}]{Gour22}%
  \BibitemOpen
  \bibfield  {author} {\bibinfo {author} {\bibfnamefont {G.}~\bibnamefont {Gour}},\ }\bibfield  {title} {\bibinfo {title} {Role of quantum coherence in thermodynamics},\ }\href {https://doi.org/10.1103/PRXQuantum.3.040323} {\bibfield  {journal} {\bibinfo  {journal} {PRX Quantum}\ }\textbf {\bibinfo {volume} {3}},\ \bibinfo {pages} {040323} (\bibinfo {year} {2022})}\BibitemShut {NoStop}%
\bibitem [{\citenamefont {Wootters}\ and\ \citenamefont {Zurek}(1982)}]{wooterzurek82}%
  \BibitemOpen
  \bibfield  {author} {\bibinfo {author} {\bibfnamefont {W.~K.}\ \bibnamefont {Wootters}}\ and\ \bibinfo {author} {\bibfnamefont {W.~H.}\ \bibnamefont {Zurek}},\ }\bibfield  {title} {\bibinfo {title} {A single quantum cannot be cloned},\ }\href {https://doi.org/10.1038/299802a0} {\bibfield  {journal} {\bibinfo  {journal} {Nature}\ }\textbf {\bibinfo {volume} {299}},\ \bibinfo {pages} {802} (\bibinfo {year} {1982})}\BibitemShut {NoStop}%
\bibitem [{\citenamefont {Werner}(1998)}]{Werner1998}%
  \BibitemOpen
  \bibfield  {author} {\bibinfo {author} {\bibfnamefont {R.~F.}\ \bibnamefont {Werner}},\ }\bibfield  {title} {\bibinfo {title} {Optimal cloning of pure states},\ }\href {https://doi.org/10.1103/PhysRevA.58.1827} {\bibfield  {journal} {\bibinfo  {journal} {Phys. Rev. A}\ }\textbf {\bibinfo {volume} {58}},\ \bibinfo {pages} {1827} (\bibinfo {year} {1998})},\ \Eprint {https://arxiv.org/abs/quant-ph/9804001} {arXiv:quant-ph/9804001} \BibitemShut {NoStop}%
\bibitem [{\citenamefont {Kalev}\ and\ \citenamefont {Hen}(2008)}]{CalevHen2008}%
  \BibitemOpen
  \bibfield  {author} {\bibinfo {author} {\bibfnamefont {A.}~\bibnamefont {Kalev}}\ and\ \bibinfo {author} {\bibfnamefont {I.}~\bibnamefont {Hen}},\ }\bibfield  {title} {\bibinfo {title} {{No-Broadcasting Theorem and Its Classical Counterpart}},\ }\href {https://doi.org/10.1103/PhysRevLett.100.210502} {\bibfield  {journal} {\bibinfo  {journal} {Phys. Rev. Lett.}\ }\textbf {\bibinfo {volume} {100}},\ \bibinfo {pages} {210502} (\bibinfo {year} {2008})},\ \Eprint {https://arxiv.org/abs/0704.1754} {arXiv:0704.1754} \BibitemShut {NoStop}%
\bibitem [{\citenamefont {Daffertshofer}\ \emph {et~al.}(2002)\citenamefont {Daffertshofer}, \citenamefont {Plastino},\ and\ \citenamefont {Plastino}}]{daffertshofer2002classical}%
  \BibitemOpen
  \bibfield  {author} {\bibinfo {author} {\bibfnamefont {A.}~\bibnamefont {Daffertshofer}}, \bibinfo {author} {\bibfnamefont {A.}~\bibnamefont {Plastino}},\ and\ \bibinfo {author} {\bibfnamefont {A.}~\bibnamefont {Plastino}},\ }\bibfield  {title} {\bibinfo {title} {Classical no-cloning theorem},\ }\href {https://doi.org/10.1103/PhysRevLett.88.210601} {\bibfield  {journal} {\bibinfo  {journal} {Physical review letters}\ }\textbf {\bibinfo {volume} {88}},\ \bibinfo {pages} {210601} (\bibinfo {year} {2002})}\BibitemShut {NoStop}%
\bibitem [{\citenamefont {Debarba}\ \emph {et~al.}(2024)\citenamefont {Debarba}, \citenamefont {Huber},\ and\ \citenamefont {Friis}}]{debarba2024}%
  \BibitemOpen
  \bibfield  {author} {\bibinfo {author} {\bibfnamefont {T.}~\bibnamefont {Debarba}}, \bibinfo {author} {\bibfnamefont {M.}~\bibnamefont {Huber}},\ and\ \bibinfo {author} {\bibfnamefont {N.}~\bibnamefont {Friis}},\ }\bibfield  {title} {\bibinfo {title} {Broadcasting quantum information using finite resources},\ }\href {https://arxiv.org/abs/2403.07660} {\bibfield  {journal} {\bibinfo  {journal} {arXiv:2403.07660}\ } (\bibinfo {year} {2024})}\BibitemShut {NoStop}%
\bibitem [{\citenamefont {Guryanova}\ \emph {et~al.}(2020{\natexlab{a}})\citenamefont {Guryanova}, \citenamefont {Friis},\ and\ \citenamefont {Huber}}]{GuryanovaFriisHuber2018}%
  \BibitemOpen
  \bibfield  {author} {\bibinfo {author} {\bibfnamefont {Y.}~\bibnamefont {Guryanova}}, \bibinfo {author} {\bibfnamefont {N.}~\bibnamefont {Friis}},\ and\ \bibinfo {author} {\bibfnamefont {M.}~\bibnamefont {Huber}},\ }\bibfield  {title} {\bibinfo {title} {Ideal projective measurements have infinite resource costs},\ }\href {https://doi.org/10.22331/q-2020-01-13-222} {\bibfield  {journal} {\bibinfo  {journal} {Quantum}\ }\textbf {\bibinfo {volume} {4}},\ \bibinfo {pages} {222} (\bibinfo {year} {2020}{\natexlab{a}})},\ \Eprint {https://arxiv.org/abs/1805.11899} {arXiv:1805.11899} \BibitemShut {NoStop}%
\bibitem [{\citenamefont {Bilokur}\ \emph {et~al.}(2024)\citenamefont {Bilokur}, \citenamefont {Gopalakrishnan},\ and\ \citenamefont {Majidy}}]{qec_landauer}%
  \BibitemOpen
  \bibfield  {author} {\bibinfo {author} {\bibfnamefont {M.}~\bibnamefont {Bilokur}}, \bibinfo {author} {\bibfnamefont {S.}~\bibnamefont {Gopalakrishnan}},\ and\ \bibinfo {author} {\bibfnamefont {S.}~\bibnamefont {Majidy}},\ }\bibfield  {title} {\bibinfo {title} {Thermodynamic limitations on fault-tolerant quantum computing},\ }\href {https://doi.org/10.48550/arXiv.2411.12805} {\bibfield  {journal} {\bibinfo  {journal} {arXiv preprint}\ } (\bibinfo {year} {2024})},\ \Eprint {https://arxiv.org/abs/2411.12805} {arXiv:2411.12805 [quant-ph]} \BibitemShut {NoStop}%
\bibitem [{\citenamefont {Li}\ \emph {et~al.}(2015)\citenamefont {Li}, \citenamefont {Humphreys}, \citenamefont {Mendoza},\ and\ \citenamefont {Benjamin}}]{Li15}%
  \BibitemOpen
  \bibfield  {author} {\bibinfo {author} {\bibfnamefont {Y.}~\bibnamefont {Li}}, \bibinfo {author} {\bibfnamefont {P.~C.}\ \bibnamefont {Humphreys}}, \bibinfo {author} {\bibfnamefont {G.~J.}\ \bibnamefont {Mendoza}},\ and\ \bibinfo {author} {\bibfnamefont {S.~C.}\ \bibnamefont {Benjamin}},\ }\bibfield  {title} {\bibinfo {title} {Resource costs for fault-tolerant linear optical quantum computing},\ }\href {https://doi.org/10.1103/PhysRevX.5.041007} {\bibfield  {journal} {\bibinfo  {journal} {Phys. Rev. X}\ }\textbf {\bibinfo {volume} {5}},\ \bibinfo {pages} {041007} (\bibinfo {year} {2015})}\BibitemShut {NoStop}%
\bibitem [{\citenamefont {Fellous-Asiani}\ \emph {et~al.}(2021)\citenamefont {Fellous-Asiani}, \citenamefont {Chai}, \citenamefont {Whitney}, \citenamefont {Auff\`eves},\ and\ \citenamefont {Ng}}]{Fellous21}%
  \BibitemOpen
  \bibfield  {author} {\bibinfo {author} {\bibfnamefont {M.}~\bibnamefont {Fellous-Asiani}}, \bibinfo {author} {\bibfnamefont {J.~H.}\ \bibnamefont {Chai}}, \bibinfo {author} {\bibfnamefont {R.~S.}\ \bibnamefont {Whitney}}, \bibinfo {author} {\bibfnamefont {A.}~\bibnamefont {Auff\`eves}},\ and\ \bibinfo {author} {\bibfnamefont {H.~K.}\ \bibnamefont {Ng}},\ }\bibfield  {title} {\bibinfo {title} {Limitations in quantum computing from resource constraints},\ }\href {https://doi.org/10.1103/PRXQuantum.2.040335} {\bibfield  {journal} {\bibinfo  {journal} {PRX Quantum}\ }\textbf {\bibinfo {volume} {2}},\ \bibinfo {pages} {040335} (\bibinfo {year} {2021})}\BibitemShut {NoStop}%
\bibitem [{\citenamefont {Danageozian}\ \emph {et~al.}(2022)\citenamefont {Danageozian}, \citenamefont {Wilde},\ and\ \citenamefont {Buscemi}}]{Danageozian22}%
  \BibitemOpen
  \bibfield  {author} {\bibinfo {author} {\bibfnamefont {A.}~\bibnamefont {Danageozian}}, \bibinfo {author} {\bibfnamefont {M.~M.}\ \bibnamefont {Wilde}},\ and\ \bibinfo {author} {\bibfnamefont {F.}~\bibnamefont {Buscemi}},\ }\bibfield  {title} {\bibinfo {title} {Thermodynamic constraints on quantum information gain and error correction: A triple trade-off},\ }\href {https://link.aps.org/doi/10.1103/PRXQuantum.3.020318} {\bibfield  {journal} {\bibinfo  {journal} {PRX Quantum}\ }\textbf {\bibinfo {volume} {3}},\ \bibinfo {pages} {020318} (\bibinfo {year} {2022})}\BibitemShut {NoStop}%
\bibitem [{\citenamefont {Horodecki}\ and\ \citenamefont {Oppenheim}(2013{\natexlab{a}})}]{irreversibility_coherence}%
  \BibitemOpen
  \bibfield  {author} {\bibinfo {author} {\bibfnamefont {M.}~\bibnamefont {Horodecki}}\ and\ \bibinfo {author} {\bibfnamefont {J.}~\bibnamefont {Oppenheim}},\ }\bibfield  {title} {\bibinfo {title} {Fundamental limitations for quantum and nanoscale thermodynamics},\ }\href {https://doi.org/10.1038/ncomms3059} {\bibfield  {journal} {\bibinfo  {journal} {Nature Communications}\ }\textbf {\bibinfo {volume} {4}},\ \bibinfo {pages} {2059} (\bibinfo {year} {2013}{\natexlab{a}})}\BibitemShut {NoStop}%
\bibitem [{\citenamefont {Zhang}\ \emph {et~al.}(2025)\citenamefont {Zhang}, \citenamefont {Zhang}, \citenamefont {Sun}, \citenamefont {Lin}, \citenamefont {Huang}, \citenamefont {Lv},\ and\ \citenamefont {Yuan}}]{zhang2025}%
  \BibitemOpen
  \bibfield  {author} {\bibinfo {author} {\bibfnamefont {Y.}~\bibnamefont {Zhang}}, \bibinfo {author} {\bibfnamefont {X.}~\bibnamefont {Zhang}}, \bibinfo {author} {\bibfnamefont {J.}~\bibnamefont {Sun}}, \bibinfo {author} {\bibfnamefont {H.}~\bibnamefont {Lin}}, \bibinfo {author} {\bibfnamefont {Y.}~\bibnamefont {Huang}}, \bibinfo {author} {\bibfnamefont {D.}~\bibnamefont {Lv}},\ and\ \bibinfo {author} {\bibfnamefont {X.}~\bibnamefont {Yuan}},\ }\href {https://doi.org/10.48550/arXiv.2502.02139} {\bibinfo {title} {Fault-tolerant quantum algorithms for quantum molecular systems: A survey}} (\bibinfo {year} {2025}),\ \bibinfo {note} {arXiv:2502.02139 [quant-ph], submitted on 4 Feb 2025},\ \Eprint {https://arxiv.org/abs/2502.02139} {arXiv:2502.02139 [quant-ph]} \BibitemShut {NoStop}%
\bibitem [{\citenamefont {Kshirsagar}\ \emph {et~al.}(2024)\citenamefont {Kshirsagar}, \citenamefont {Katabarwa},\ and\ \citenamefont {Johnson}}]{Kshirsagar2024}%
  \BibitemOpen
  \bibfield  {author} {\bibinfo {author} {\bibfnamefont {R.}~\bibnamefont {Kshirsagar}}, \bibinfo {author} {\bibfnamefont {A.}~\bibnamefont {Katabarwa}},\ and\ \bibinfo {author} {\bibfnamefont {P.~D.}\ \bibnamefont {Johnson}},\ }\bibfield  {title} {\bibinfo {title} {On proving the robustness of algorithms for early fault-tolerant quantum computers},\ }\href {https://doi.org/10.22331/q-2024-11-20-1531} {\bibfield  {journal} {\bibinfo  {journal} {{Quantum}}\ }\textbf {\bibinfo {volume} {8}},\ \bibinfo {pages} {1531} (\bibinfo {year} {2024})}\BibitemShut {NoStop}%
\bibitem [{\citenamefont {Biswas}\ \emph {et~al.}(2022)\citenamefont {Biswas}, \citenamefont {Junior}, \citenamefont {Horodecki},\ and\ \citenamefont {Korzekwa}}]{Biswas22}%
  \BibitemOpen
  \bibfield  {author} {\bibinfo {author} {\bibfnamefont {T.}~\bibnamefont {Biswas}}, \bibinfo {author} {\bibfnamefont {A.~d.~O.}\ \bibnamefont {Junior}}, \bibinfo {author} {\bibfnamefont {M.}~\bibnamefont {Horodecki}},\ and\ \bibinfo {author} {\bibfnamefont {K.}~\bibnamefont {Korzekwa}},\ }\bibfield  {title} {\bibinfo {title} {Fluctuation-dissipation relations for thermodynamic distillation processes},\ }\href {https://doi.org/10.1103/PhysRevE.105.054127} {\bibfield  {journal} {\bibinfo  {journal} {Phys. Rev. E}\ }\textbf {\bibinfo {volume} {105}},\ \bibinfo {pages} {054127} (\bibinfo {year} {2022})}\BibitemShut {NoStop}%
\bibitem [{\citenamefont {Schumacher}\ and\ \citenamefont {Westmoreland}(1997)}]{schumacher1997sending}%
  \BibitemOpen
  \bibfield  {author} {\bibinfo {author} {\bibfnamefont {B.}~\bibnamefont {Schumacher}}\ and\ \bibinfo {author} {\bibfnamefont {M.~D.}\ \bibnamefont {Westmoreland}},\ }\bibfield  {title} {\bibinfo {title} {Sending classical information via noisy quantum channels},\ }\href {https://doi.org/10.1103/PhysRevA.56.131} {\bibfield  {journal} {\bibinfo  {journal} {Physical Review A}\ }\textbf {\bibinfo {volume} {56}},\ \bibinfo {pages} {131} (\bibinfo {year} {1997})}\BibitemShut {NoStop}%
\bibitem [{\citenamefont {Horodecki}(1998)}]{horodecki1998limits}%
  \BibitemOpen
  \bibfield  {author} {\bibinfo {author} {\bibfnamefont {M.}~\bibnamefont {Horodecki}},\ }\bibfield  {title} {\bibinfo {title} {Limits for compression of quantum information carried by ensembles of mixed states},\ }\href {https://doi.org/10.1103/PhysRevA.57.3364} {\bibfield  {journal} {\bibinfo  {journal} {Physical Review A}\ }\textbf {\bibinfo {volume} {57}},\ \bibinfo {pages} {3364} (\bibinfo {year} {1998})}\BibitemShut {NoStop}%
\bibitem [{\citenamefont {Holevo}(1973)}]{Holevo1973}%
  \BibitemOpen
  \bibfield  {author} {\bibinfo {author} {\bibfnamefont {A.~S.}\ \bibnamefont {Holevo}},\ }\bibfield  {title} {\bibinfo {title} {Bounds for the quantity of information transmitted by a quantum communication channel},\ }\href@noop {} {\bibfield  {journal} {\bibinfo  {journal} {Problems Inform. Transmission}\ }\textbf {\bibinfo {volume} {9}},\ \bibinfo {pages} {177} (\bibinfo {year} {1973})},\ \bibinfo {note} {url: \href{http://mi.mathnet.ru/eng/ppi903}{http://mi.mathnet.ru/eng/ppi903 }}\BibitemShut {NoStop}%
\bibitem [{\citenamefont {Chiribella}\ \emph {et~al.}(2022)\citenamefont {Chiribella}, \citenamefont {Meng}, \citenamefont {Renner},\ and\ \citenamefont {Yung}}]{chiribella22}%
  \BibitemOpen
  \bibfield  {author} {\bibinfo {author} {\bibfnamefont {G.}~\bibnamefont {Chiribella}}, \bibinfo {author} {\bibfnamefont {F.}~\bibnamefont {Meng}}, \bibinfo {author} {\bibfnamefont {R.}~\bibnamefont {Renner}},\ and\ \bibinfo {author} {\bibfnamefont {M.-H.}\ \bibnamefont {Yung}},\ }\bibfield  {title} {\bibinfo {title} {The nonequilibrium cost of accurate information processing},\ }\href {https://doi.org/10.1038/s41467-022-34541-w} {\bibfield  {journal} {\bibinfo  {journal} {Nature Communications}\ }\textbf {\bibinfo {volume} {13}},\ \bibinfo {pages} {7155} (\bibinfo {year} {2022})}\BibitemShut {NoStop}%
\bibitem [{\citenamefont {Narasimhachar}\ \emph {et~al.}(2019)\citenamefont {Narasimhachar}, \citenamefont {Thompson}, \citenamefont {Ma}, \citenamefont {Gour},\ and\ \citenamefont {Gu}}]{Narasimhachar2019}%
  \BibitemOpen
  \bibfield  {author} {\bibinfo {author} {\bibfnamefont {V.}~\bibnamefont {Narasimhachar}}, \bibinfo {author} {\bibfnamefont {J.}~\bibnamefont {Thompson}}, \bibinfo {author} {\bibfnamefont {J.}~\bibnamefont {Ma}}, \bibinfo {author} {\bibfnamefont {G.}~\bibnamefont {Gour}},\ and\ \bibinfo {author} {\bibfnamefont {M.}~\bibnamefont {Gu}},\ }\bibfield  {title} {\bibinfo {title} {Quantifying memory capacity as a quantum thermodynamic resource},\ }\href {https://doi.org/10.1103/PhysRevLett.122.060601} {\bibfield  {journal} {\bibinfo  {journal} {Physical Review Letters}\ }\textbf {\bibinfo {volume} {122}},\ \bibinfo {pages} {060601} (\bibinfo {year} {2019})}\BibitemShut {NoStop}%
\bibitem [{\citenamefont {Hsieh}\ and\ \citenamefont {Chen}(2025)}]{Hsieh2025PRL}%
  \BibitemOpen
  \bibfield  {author} {\bibinfo {author} {\bibfnamefont {C.-Y.}\ \bibnamefont {Hsieh}}\ and\ \bibinfo {author} {\bibfnamefont {S.-L.}\ \bibnamefont {Chen}},\ }\bibfield  {title} {\bibinfo {title} {Dynamical landauer principle: Quantifying information transmission by thermodynamics},\ }\href {https://doi.org/10.1103/PhysRevLett.134.050404} {\bibfield  {journal} {\bibinfo  {journal} {Physical Review Letters}\ }\textbf {\bibinfo {volume} {134}},\ \bibinfo {pages} {050404} (\bibinfo {year} {2025})}\BibitemShut {NoStop}%
\bibitem [{\citenamefont {Hsieh}(2025)}]{Hsieh2025PRA}%
  \BibitemOpen
  \bibfield  {author} {\bibinfo {author} {\bibfnamefont {C.-Y.}\ \bibnamefont {Hsieh}},\ }\bibfield  {title} {\bibinfo {title} {Dynamical landauer principle: Thermodynamic criteria of transmitting classical information},\ }\href {https://doi.org/10.1103/PhysRevA.111.022207} {\bibfield  {journal} {\bibinfo  {journal} {Physical Review A}\ }\textbf {\bibinfo {volume} {111}},\ \bibinfo {pages} {022207} (\bibinfo {year} {2025})}\BibitemShut {NoStop}%
\bibitem [{\citenamefont {Xuereb}\ \emph {et~al.}(2025)\citenamefont {Xuereb}, \citenamefont {Debarba}, \citenamefont {Huber},\ and\ \citenamefont {Erker}}]{Xuereb_2025}%
  \BibitemOpen
  \bibfield  {author} {\bibinfo {author} {\bibfnamefont {J.}~\bibnamefont {Xuereb}}, \bibinfo {author} {\bibfnamefont {T.}~\bibnamefont {Debarba}}, \bibinfo {author} {\bibfnamefont {M.}~\bibnamefont {Huber}},\ and\ \bibinfo {author} {\bibfnamefont {P.}~\bibnamefont {Erker}},\ }\bibfield  {title} {\bibinfo {title} {{Quantum coding with finite thermodynamic resources}},\ }\href {https://doi.org/10.1088/2058-9565/adb0e9} {\bibfield  {journal} {\bibinfo  {journal} {Quantum Science and Technology}\ }\textbf {\bibinfo {volume} {10}},\ \bibinfo {pages} {25030} (\bibinfo {year} {2025})}\BibitemShut {NoStop}%
\bibitem [{\citenamefont {Schr{\"o}dinger}(1935)}]{schrodinger1935mixture}%
  \BibitemOpen
  \bibfield  {author} {\bibinfo {author} {\bibfnamefont {E.}~\bibnamefont {Schr{\"o}dinger}},\ }\bibfield  {title} {\bibinfo {title} {Discussion of probability relations between separated systems},\ }\href {https://doi.org/10.1017/S0305004100013554} {\bibfield  {journal} {\bibinfo  {journal} {Math. Proc. Camb. Philos. Soc.}\ }\textbf {\bibinfo {volume} {31}},\ \bibinfo {pages} {555} (\bibinfo {year} {1935})}\BibitemShut {NoStop}%
\bibitem [{\citenamefont {Hughston}\ \emph {et~al.}(1993)\citenamefont {Hughston}, \citenamefont {Jozsa},\ and\ \citenamefont {Wootters}}]{HJW93}%
  \BibitemOpen
  \bibfield  {author} {\bibinfo {author} {\bibfnamefont {L.~P.}\ \bibnamefont {Hughston}}, \bibinfo {author} {\bibfnamefont {R.}~\bibnamefont {Jozsa}},\ and\ \bibinfo {author} {\bibfnamefont {W.~K.}\ \bibnamefont {Wootters}},\ }\bibfield  {title} {\bibinfo {title} {A complete classification of quantum ensembles having a given density matrix},\ }\href {https://doi.org/https://doi.org/10.1016/0375-9601(93)90880-9} {\bibfield  {journal} {\bibinfo  {journal} {Physics Letters A}\ }\textbf {\bibinfo {volume} {183}},\ \bibinfo {pages} {14} (\bibinfo {year} {1993})}\BibitemShut {NoStop}%
\bibitem [{\citenamefont {Brand\~{a}o}\ \emph {et~al.}(2013)\citenamefont {Brand\~{a}o}, \citenamefont {Horodecki}, \citenamefont {Oppenheim}, \citenamefont {Renes},\ and\ \citenamefont {Spekkens}}]{Brandao13}%
  \BibitemOpen
  \bibfield  {author} {\bibinfo {author} {\bibfnamefont {F.~G. S.~L.}\ \bibnamefont {Brand\~{a}o}}, \bibinfo {author} {\bibfnamefont {M.}~\bibnamefont {Horodecki}}, \bibinfo {author} {\bibfnamefont {J.}~\bibnamefont {Oppenheim}}, \bibinfo {author} {\bibfnamefont {J.~M.}\ \bibnamefont {Renes}},\ and\ \bibinfo {author} {\bibfnamefont {R.~W.}\ \bibnamefont {Spekkens}},\ }\bibfield  {title} {\bibinfo {title} {Resource theory of quantum states out of thermal equilibrium},\ }\href {https://doi.org/10.1103/PhysRevLett.111.250404} {\bibfield  {journal} {\bibinfo  {journal} {Physical Review Letters}\ }\textbf {\bibinfo {volume} {111}},\ \bibinfo {pages} {250404} (\bibinfo {year} {2013})}\BibitemShut {NoStop}%
\bibitem [{\citenamefont {Horodecki}\ and\ \citenamefont {Oppenheim}(2013{\natexlab{b}})}]{Horodecki13}%
  \BibitemOpen
  \bibfield  {author} {\bibinfo {author} {\bibfnamefont {M.}~\bibnamefont {Horodecki}}\ and\ \bibinfo {author} {\bibfnamefont {J.}~\bibnamefont {Oppenheim}},\ }\bibfield  {title} {\bibinfo {title} {Fundamental limitations for quantum and nanoscale thermodynamics},\ }\href {https://doi.org/10.1038/ncomms3059} {\bibfield  {journal} {\bibinfo  {journal} {Nature Communications}\ }\textbf {\bibinfo {volume} {4}},\ \bibinfo {pages} {2059} (\bibinfo {year} {2013}{\natexlab{b}})}\BibitemShut {NoStop}%
\bibitem [{\citenamefont {Preskill}(2016)}]{preskill16a}%
  \BibitemOpen
  \bibfield  {author} {\bibinfo {author} {\bibfnamefont {J.}~\bibnamefont {Preskill}},\ }\bibfield  {title} {\bibinfo {title} {Quantum shannon theory},\ }\href {https://arxiv.org/abs/1604.07450} {\bibfield  {journal} {\bibinfo  {journal} {arXiv preprint arXiv:1604.07450}\ } (\bibinfo {year} {2016})}\BibitemShut {NoStop}%
\bibitem [{\citenamefont {Ollivier}\ and\ \citenamefont {Zurek}(2001)}]{OllivierZurek01}%
  \BibitemOpen
  \bibfield  {author} {\bibinfo {author} {\bibfnamefont {H.}~\bibnamefont {Ollivier}}\ and\ \bibinfo {author} {\bibfnamefont {W.~H.}\ \bibnamefont {Zurek}},\ }\bibfield  {title} {\bibinfo {title} {Quantum discord: A measure of the quantumness of correlations},\ }\href {https://doi.org/10.1103/PhysRevLett.88.017901} {\bibfield  {journal} {\bibinfo  {journal} {Phys. Rev. Lett.}\ }\textbf {\bibinfo {volume} {88}},\ \bibinfo {pages} {017901} (\bibinfo {year} {2001})}\BibitemShut {NoStop}%
\bibitem [{\citenamefont {Ivanovic}(1987)}]{ivanovic1987}%
  \BibitemOpen
  \bibfield  {author} {\bibinfo {author} {\bibfnamefont {I.}~\bibnamefont {Ivanovic}},\ }\bibfield  {title} {\bibinfo {title} {How to differentiate between non-orthogonal states},\ }\href {https://www.sciencedirect.com/science/article/pii/0375960187902222} {\bibfield  {journal} {\bibinfo  {journal} {Physics Letters A}\ }\textbf {\bibinfo {volume} {123}},\ \bibinfo {pages} {257} (\bibinfo {year} {1987})}\BibitemShut {NoStop}%
\bibitem [{\citenamefont {Chefles}(2000)}]{Chefles2000}%
  \BibitemOpen
  \bibfield  {author} {\bibinfo {author} {\bibfnamefont {A.}~\bibnamefont {Chefles}},\ }\bibfield  {title} {\bibinfo {title} {Quantum state discrimination},\ }\href {https://doi.org/10.1080/00107510010002599} {\bibfield  {journal} {\bibinfo  {journal} {Contemporary Physics}\ }\textbf {\bibinfo {volume} {41}},\ \bibinfo {pages} {401} (\bibinfo {year} {2000})}\BibitemShut {NoStop}%
\bibitem [{\citenamefont {Roa}\ \emph {et~al.}(2002)\citenamefont {Roa}, \citenamefont {Retamal},\ and\ \citenamefont {Saavedra}}]{Roa2002}%
  \BibitemOpen
  \bibfield  {author} {\bibinfo {author} {\bibfnamefont {L.}~\bibnamefont {Roa}}, \bibinfo {author} {\bibfnamefont {J.~C.}\ \bibnamefont {Retamal}},\ and\ \bibinfo {author} {\bibfnamefont {C.}~\bibnamefont {Saavedra}},\ }\bibfield  {title} {\bibinfo {title} {Quantum-state discrimination},\ }\href {https://link.aps.org/doi/10.1103/PhysRevA.66.012103} {\bibfield  {journal} {\bibinfo  {journal} {Phys. Rev. A}\ }\textbf {\bibinfo {volume} {66}},\ \bibinfo {pages} {012103} (\bibinfo {year} {2002})}\BibitemShut {NoStop}%
\bibitem [{\citenamefont {Kraus}(1971)}]{kraus1971}%
  \BibitemOpen
  \bibfield  {author} {\bibinfo {author} {\bibfnamefont {K.}~\bibnamefont {Kraus}},\ }\bibfield  {title} {\bibinfo {title} {General state changes in quantum theory},\ }\href {https://doi.org/https://doi.org/10.1016/0003-4916(71)90108-4} {\bibfield  {journal} {\bibinfo  {journal} {Annals of Physics}\ }\textbf {\bibinfo {volume} {64}},\ \bibinfo {pages} {311} (\bibinfo {year} {1971})}\BibitemShut {NoStop}%
\bibitem [{\citenamefont {Choi}(1975)}]{choi1975}%
  \BibitemOpen
  \bibfield  {author} {\bibinfo {author} {\bibfnamefont {M.-D.}\ \bibnamefont {Choi}},\ }\bibfield  {title} {\bibinfo {title} {Completely positive linear maps on complex matrices},\ }\href {https://doi.org/https://doi.org/10.1016/0024-3795(75)90075-0} {\bibfield  {journal} {\bibinfo  {journal} {Linear algebra and its applications}\ }\textbf {\bibinfo {volume} {10}},\ \bibinfo {pages} {285} (\bibinfo {year} {1975})}\BibitemShut {NoStop}%
\bibitem [{\citenamefont {Faist}\ and\ \citenamefont {Renner}(2018)}]{FaistRenner2018}%
  \BibitemOpen
  \bibfield  {author} {\bibinfo {author} {\bibfnamefont {P.}~\bibnamefont {Faist}}\ and\ \bibinfo {author} {\bibfnamefont {R.}~\bibnamefont {Renner}},\ }\bibfield  {title} {\bibinfo {title} {{Fundamental Work Cost of Quantum Processes}},\ }\href {https://doi.org/10.1103/PhysRevX.8.021011} {\bibfield  {journal} {\bibinfo  {journal} {Phys. Rev. X}\ }\textbf {\bibinfo {volume} {8}},\ \bibinfo {pages} {021011} (\bibinfo {year} {2018})},\ \Eprint {https://arxiv.org/abs/1709.00506} {arXiv:1709.00506} \BibitemShut {NoStop}%
\bibitem [{\citenamefont {Stinespring}(1955)}]{Stinespring55}%
  \BibitemOpen
  \bibfield  {author} {\bibinfo {author} {\bibfnamefont {W.~F.}\ \bibnamefont {Stinespring}},\ }\bibfield  {title} {\bibinfo {title} {Positive functions on {C}$^*$-algebras},\ }\href {https://doi.org/10.2307/2032342} {\bibfield  {journal} {\bibinfo  {journal} {Proceedings of the American Mathematical Society}\ }\textbf {\bibinfo {volume} {6}},\ \bibinfo {pages} {211} (\bibinfo {year} {1955})}\BibitemShut {NoStop}%
\bibitem [{\citenamefont {Cadney}\ \emph {et~al.}(2013)\citenamefont {Cadney}, \citenamefont {Huber}, \citenamefont {Linden},\ and\ \citenamefont {Winter}}]{Cadney2013}%
  \BibitemOpen
  \bibfield  {author} {\bibinfo {author} {\bibfnamefont {J.}~\bibnamefont {Cadney}}, \bibinfo {author} {\bibfnamefont {M.}~\bibnamefont {Huber}}, \bibinfo {author} {\bibfnamefont {N.}~\bibnamefont {Linden}},\ and\ \bibinfo {author} {\bibfnamefont {A.}~\bibnamefont {Winter}},\ }\bibfield  {title} {\bibinfo {title} {Inequalities for the ranks of quantum states},\ }\href {https://doi.org/10.48550/arXiv.1308.0539} {\bibfield  {journal} {\bibinfo  {journal} {arXiv preprint}\ } (\bibinfo {year} {2013})},\ \bibinfo {note} {arXiv:1308.0539 [quant-ph]},\ \Eprint {https://arxiv.org/abs/1308.0539} {1308.0539} \BibitemShut {NoStop}%
\bibitem [{\citenamefont {Bengtsson}\ and\ \citenamefont {Zyczkowski}(2017)}]{zycbook}%
  \BibitemOpen
  \bibfield  {author} {\bibinfo {author} {\bibfnamefont {I.}~\bibnamefont {Bengtsson}}\ and\ \bibinfo {author} {\bibfnamefont {K.}~\bibnamefont {Zyczkowski}},\ }\href {https://doi.org/10.1017/CBO9780511535048} {\bibinfo {title} {Geometry of quantum states: an introduction to quantum entanglement}} (\bibinfo {year} {2017})\BibitemShut {NoStop}%
\bibitem [{\citenamefont {Barnett}\ and\ \citenamefont {Croke}(2009)}]{barnett2009conditions}%
  \BibitemOpen
  \bibfield  {author} {\bibinfo {author} {\bibfnamefont {S.~M.}\ \bibnamefont {Barnett}}\ and\ \bibinfo {author} {\bibfnamefont {S.}~\bibnamefont {Croke}},\ }\bibfield  {title} {\bibinfo {title} {On the conditions for discrimination between quantum states with minimum error},\ }\href {https://dx.doi.org/10.1088/1751-8113/42/6/062001} {\bibfield  {journal} {\bibinfo  {journal} {Journal of Physics. A, Mathematical and Theoretical}\ }\textbf {\bibinfo {volume} {42}} (\bibinfo {year} {2009})}\BibitemShut {NoStop}%
\bibitem [{\citenamefont {{Debarba}}\ \emph {et~al.}(2019)\citenamefont {{Debarba}}, \citenamefont {{Manzano}}, \citenamefont {{Guryanova}}, \citenamefont {{Huber}},\ and\ \citenamefont {{Friis}}}]{DebarbaManzanoGuryanovaHuberFriis2019}%
  \BibitemOpen
  \bibfield  {author} {\bibinfo {author} {\bibfnamefont {T.}~\bibnamefont {{Debarba}}}, \bibinfo {author} {\bibfnamefont {G.}~\bibnamefont {{Manzano}}}, \bibinfo {author} {\bibfnamefont {Y.}~\bibnamefont {{Guryanova}}}, \bibinfo {author} {\bibfnamefont {M.}~\bibnamefont {{Huber}}},\ and\ \bibinfo {author} {\bibfnamefont {N.}~\bibnamefont {{Friis}}},\ }\bibfield  {title} {\bibinfo {title} {{Work estimation and work fluctuations in the presence of non-ideal measurements}},\ }\href {https://doi.org/10.1088/1367-2630/ab4d9d} {\bibfield  {journal} {\bibinfo  {journal} {New J. Phys.}\ }\textbf {\bibinfo {volume} {21}},\ \bibinfo {pages} {113002} (\bibinfo {year} {2019})},\ \Eprint {https://arxiv.org/abs/1902.08568} {arXiv:1902.08568} \BibitemShut {NoStop}%
\bibitem [{\citenamefont {Fano}\ and\ \citenamefont {Hawkins}(1961)}]{fano1961transmission}%
  \BibitemOpen
  \bibfield  {author} {\bibinfo {author} {\bibfnamefont {R.~M.}\ \bibnamefont {Fano}}\ and\ \bibinfo {author} {\bibfnamefont {D.}~\bibnamefont {Hawkins}},\ }\bibfield  {title} {\bibinfo {title} {Transmission of information: A statistical theory of communications},\ }\href {https://doi.org/10.1119/1.1937609} {\bibfield  {journal} {\bibinfo  {journal} {American Journal of Physics}\ }\textbf {\bibinfo {volume} {29}},\ \bibinfo {pages} {793} (\bibinfo {year} {1961})}\BibitemShut {NoStop}%
\bibitem [{\citenamefont {Reeb}\ and\ \citenamefont {Wolf}(2014)}]{reebandwolf2014}%
  \BibitemOpen
  \bibfield  {author} {\bibinfo {author} {\bibfnamefont {D.}~\bibnamefont {Reeb}}\ and\ \bibinfo {author} {\bibfnamefont {M.~M.}\ \bibnamefont {Wolf}},\ }\bibfield  {title} {\bibinfo {title} {{An improved Landauer Principle with finite-size corrections}},\ }\href {https://doi.org/10.1088/1367-2630/16/10/103011} {\bibfield  {journal} {\bibinfo  {journal} {New J. Phys.}\ }\textbf {\bibinfo {volume} {16}},\ \bibinfo {pages} {103011} (\bibinfo {year} {2014})},\ \Eprint {https://arxiv.org/abs/1306.4352} {arXiv:1306.4352} \BibitemShut {NoStop}%
\bibitem [{\citenamefont {Donald}(1987)}]{donald1987}%
  \BibitemOpen
  \bibfield  {author} {\bibinfo {author} {\bibfnamefont {M.~J.}\ \bibnamefont {Donald}},\ }\bibfield  {title} {\bibinfo {title} {Free energy and the relative entropy},\ }\href {https://doi.org/10.1007/BF01009955} {\bibfield  {journal} {\bibinfo  {journal} {J. Stat. Phys.}\ }\textbf {\bibinfo {volume} {49}},\ \bibinfo {pages} {81} (\bibinfo {year} {1987})}\BibitemShut {NoStop}%
\bibitem [{\citenamefont {Junior}\ \emph {et~al.}(2025)\citenamefont {Junior}, \citenamefont {Brask},\ and\ \citenamefont {Chaves}}]{Junior2025}%
  \BibitemOpen
  \bibfield  {author} {\bibinfo {author} {\bibfnamefont {A.~D.~O.}\ \bibnamefont {Junior}}, \bibinfo {author} {\bibfnamefont {J.~B.}\ \bibnamefont {Brask}},\ and\ \bibinfo {author} {\bibfnamefont {R.}~\bibnamefont {Chaves}},\ }\bibfield  {title} {\bibinfo {title} {A friendly guide to exorcising maxwell's demon},\ }\href {https://arxiv.org/abs/2503.07740} {\bibfield  {journal} {\bibinfo  {journal} {arXiv:2503.07740v1}\ } (\bibinfo {year} {2025})}\BibitemShut {NoStop}%
\bibitem [{\citenamefont {Yadav}\ and\ \citenamefont {Wolpert}(2024)}]{YadavWolpert2024}%
  \BibitemOpen
  \bibfield  {author} {\bibinfo {author} {\bibfnamefont {A.}~\bibnamefont {Yadav}}\ and\ \bibinfo {author} {\bibfnamefont {D.~H.}\ \bibnamefont {Wolpert}},\ }\bibfield  {title} {\bibinfo {title} {Minimal thermodynamic cost of communication},\ }\href {https://doi.org/10.48550/arXiv.2410.14920} {\bibfield  {journal} {\bibinfo  {journal} {arXiv preprint}\ } (\bibinfo {year} {2024})},\ \bibinfo {note} {32 pages, 12 figures},\ \Eprint {https://arxiv.org/abs/2410.14920} {arXiv:2410.14920 [cond-mat.stat-mech]} \BibitemShut {NoStop}%
\bibitem [{\citenamefont {Meier}\ \emph {et~al.}(2024)\citenamefont {Meier}, \citenamefont {Huber}, \citenamefont {Erker},\ and\ \citenamefont {Xuereb}}]{meier2024autonomous}%
  \BibitemOpen
  \bibfield  {author} {\bibinfo {author} {\bibfnamefont {F.}~\bibnamefont {Meier}}, \bibinfo {author} {\bibfnamefont {M.}~\bibnamefont {Huber}}, \bibinfo {author} {\bibfnamefont {P.}~\bibnamefont {Erker}},\ and\ \bibinfo {author} {\bibfnamefont {J.}~\bibnamefont {Xuereb}},\ }\bibfield  {title} {\bibinfo {title} {Autonomous quantum processing unit: What does it take to construct a self-contained model for quantum computation?},\ }\bibfield  {journal} {\bibinfo  {journal} {arXiv preprint arXiv:2402.00111}\ }\href {https://doi.org/https://arxiv.org/abs/2402.00111} {https://arxiv.org/abs/2402.00111} (\bibinfo {year} {2024})\BibitemShut {NoStop}%
\bibitem [{\citenamefont {Faist}\ \emph {et~al.}(2015)\citenamefont {Faist}, \citenamefont {Dupuis}, \citenamefont {Oppenheim},\ and\ \citenamefont {Renner}}]{FaistDupuisOppenheimRenner2015}%
  \BibitemOpen
  \bibfield  {author} {\bibinfo {author} {\bibfnamefont {P.}~\bibnamefont {Faist}}, \bibinfo {author} {\bibfnamefont {F.}~\bibnamefont {Dupuis}}, \bibinfo {author} {\bibfnamefont {J.}~\bibnamefont {Oppenheim}},\ and\ \bibinfo {author} {\bibfnamefont {R.}~\bibnamefont {Renner}},\ }\bibfield  {title} {\bibinfo {title} {{The minimal work cost of information processing}},\ }\href {https://doi.org/10.1038/ncomms8669} {\bibfield  {journal} {\bibinfo  {journal} {Nat. Commun.}\ }\textbf {\bibinfo {volume} {6}},\ \bibinfo {pages} {7669} (\bibinfo {year} {2015})},\ \Eprint {https://arxiv.org/abs/1211.1037} {arXiv:1211.1037} \BibitemShut {NoStop}%
\bibitem [{\citenamefont {Lostaglio}(2019{\natexlab{b}})}]{Lostaglio2020}%
  \BibitemOpen
  \bibfield  {author} {\bibinfo {author} {\bibfnamefont {M.}~\bibnamefont {Lostaglio}},\ }\bibfield  {title} {\bibinfo {title} {An introductory review of the resource theory approach to thermodynamics},\ }\href {https://doi.org/10.1088/1361-6633/ab46e5} {\bibfield  {journal} {\bibinfo  {journal} {Rep. Prog. Phys.}\ }\textbf {\bibinfo {volume} {82}},\ \bibinfo {pages} {114001} (\bibinfo {year} {2019}{\natexlab{b}})}\BibitemShut {NoStop}%
\bibitem [{\citenamefont {Guryanova}\ \emph {et~al.}(2020{\natexlab{b}})\citenamefont {Guryanova}, \citenamefont {Popescu}, \citenamefont {Linden},\ and\ \citenamefont {Skrzypczyk}}]{Guryanova2020}%
  \BibitemOpen
  \bibfield  {author} {\bibinfo {author} {\bibfnamefont {Y.}~\bibnamefont {Guryanova}}, \bibinfo {author} {\bibfnamefont {S.}~\bibnamefont {Popescu}}, \bibinfo {author} {\bibfnamefont {N.}~\bibnamefont {Linden}},\ and\ \bibinfo {author} {\bibfnamefont {P.}~\bibnamefont {Skrzypczyk}},\ }\bibfield  {title} {\bibinfo {title} {Non-commutativity and macroscopic equilibrium},\ }\href {https://doi.org/10.1038/s41467-020-17649-0} {\bibfield  {journal} {\bibinfo  {journal} {Nat. Commun.}\ }\textbf {\bibinfo {volume} {11}},\ \bibinfo {pages} {3770} (\bibinfo {year} {2020}{\natexlab{b}})}\BibitemShut {NoStop}%
\bibitem [{\citenamefont {Majidy}\ \emph {et~al.}(2023{\natexlab{a}})\citenamefont {Majidy}, \citenamefont {Lasek}, \citenamefont {Huse},\ and\ \citenamefont {Yunger~Halpern}}]{Majidy23}%
  \BibitemOpen
  \bibfield  {author} {\bibinfo {author} {\bibfnamefont {S.}~\bibnamefont {Majidy}}, \bibinfo {author} {\bibfnamefont {A.}~\bibnamefont {Lasek}}, \bibinfo {author} {\bibfnamefont {D.~A.}\ \bibnamefont {Huse}},\ and\ \bibinfo {author} {\bibfnamefont {N.}~\bibnamefont {Yunger~Halpern}},\ }\bibfield  {title} {\bibinfo {title} {Non-abelian symmetry can increase entanglement entropy},\ }\href {https://doi.org/10.1103/PhysRevB.107.045102} {\bibfield  {journal} {\bibinfo  {journal} {Phys. Rev. B}\ }\textbf {\bibinfo {volume} {107}},\ \bibinfo {pages} {045102} (\bibinfo {year} {2023}{\natexlab{a}})}\BibitemShut {NoStop}%
\bibitem [{\citenamefont {Majidy}\ \emph {et~al.}(2023{\natexlab{b}})\citenamefont {Majidy}, \citenamefont {Braasch}, \citenamefont {Lasek}, \citenamefont {Upadhyaya}, \citenamefont {Kalev},\ and\ \citenamefont {Yunger~Halpern}}]{Majidy2_23}%
  \BibitemOpen
  \bibfield  {author} {\bibinfo {author} {\bibfnamefont {S.}~\bibnamefont {Majidy}}, \bibinfo {author} {\bibfnamefont {W.~F.}\ \bibnamefont {Braasch}}, \bibinfo {author} {\bibfnamefont {A.}~\bibnamefont {Lasek}}, \bibinfo {author} {\bibfnamefont {T.}~\bibnamefont {Upadhyaya}}, \bibinfo {author} {\bibfnamefont {A.}~\bibnamefont {Kalev}},\ and\ \bibinfo {author} {\bibfnamefont {N.}~\bibnamefont {Yunger~Halpern}},\ }\bibfield  {title} {\bibinfo {title} {Noncommuting conserved charges in quantum thermodynamics and beyond},\ }\href {https://doi.org/10.1038/s42254-023-00641-9} {\bibfield  {journal} {\bibinfo  {journal} {Nature Reviews Physics}\ }\textbf {\bibinfo {volume} {5}},\ \bibinfo {pages} {689} (\bibinfo {year} {2023}{\natexlab{b}})}\BibitemShut {NoStop}%
\bibitem [{\citenamefont {Barato}\ and\ \citenamefont {Seifert}(2018)}]{BaratoSeifert2018}%
  \BibitemOpen
  \bibfield  {author} {\bibinfo {author} {\bibfnamefont {A.~C.}\ \bibnamefont {Barato}}\ and\ \bibinfo {author} {\bibfnamefont {U.}~\bibnamefont {Seifert}},\ }\bibfield  {title} {\bibinfo {title} {Unifying three perspectives on information thermodynamics},\ }\href {https://doi.org/10.1088/1751-8121/aaaf3f} {\bibfield  {journal} {\bibinfo  {journal} {J. Phys. A: Math. Theor.}\ }\textbf {\bibinfo {volume} {51}},\ \bibinfo {pages} {094001} (\bibinfo {year} {2018})},\ \Eprint {https://arxiv.org/abs/1811.03988} {arXiv:1811.03988} \BibitemShut {NoStop}%
\bibitem [{\citenamefont {Misra}\ \emph {et~al.}(2016)\citenamefont {Misra}, \citenamefont {Singh}, \citenamefont {Bhattacharya},\ and\ \citenamefont {Pati}}]{misra16}%
  \BibitemOpen
  \bibfield  {author} {\bibinfo {author} {\bibfnamefont {A.}~\bibnamefont {Misra}}, \bibinfo {author} {\bibfnamefont {U.}~\bibnamefont {Singh}}, \bibinfo {author} {\bibfnamefont {S.}~\bibnamefont {Bhattacharya}},\ and\ \bibinfo {author} {\bibfnamefont {A.~K.}\ \bibnamefont {Pati}},\ }\bibfield  {title} {\bibinfo {title} {Energy cost of creating quantum coherence},\ }\href {https://doi.org/10.1103/PhysRevA.93.052335} {\bibfield  {journal} {\bibinfo  {journal} {Physical Review A}\ }\textbf {\bibinfo {volume} {93}},\ \bibinfo {pages} {052335} (\bibinfo {year} {2016})}\BibitemShut {NoStop}%
\bibitem [{\citenamefont {Francica}\ \emph {et~al.}(2020)\citenamefont {Francica}, \citenamefont {Binder}, \citenamefont {Guarnieri}, \citenamefont {Mitchison}, \citenamefont {Goold},\ and\ \citenamefont {Plastina}}]{Francica20}%
  \BibitemOpen
  \bibfield  {author} {\bibinfo {author} {\bibfnamefont {G.}~\bibnamefont {Francica}}, \bibinfo {author} {\bibfnamefont {F.~C.}\ \bibnamefont {Binder}}, \bibinfo {author} {\bibfnamefont {G.}~\bibnamefont {Guarnieri}}, \bibinfo {author} {\bibfnamefont {M.~T.}\ \bibnamefont {Mitchison}}, \bibinfo {author} {\bibfnamefont {J.}~\bibnamefont {Goold}},\ and\ \bibinfo {author} {\bibfnamefont {F.}~\bibnamefont {Plastina}},\ }\bibfield  {title} {\bibinfo {title} {Quantum coherence and ergotropy},\ }\href {https://doi.org/10.1103/PhysRevLett.125.180603} {\bibfield  {journal} {\bibinfo  {journal} {Phys. Rev. Lett.}\ }\textbf {\bibinfo {volume} {125}},\ \bibinfo {pages} {180603} (\bibinfo {year} {2020})}\BibitemShut {NoStop}%
\end{thebibliography}
\end{document}